\documentclass[11pt,journal,onecolumn]{IEEEtran}

\usepackage[english]{babel}
\usepackage{enumitem}
\usepackage[latin1]{inputenc}
\usepackage{enumerate}
\usepackage{graphicx}
\usepackage{color}
\usepackage[dvipsnames]{xcolor}
\usepackage[T1]{fontenc}
\usepackage[font=small,labelsep=space]{caption}
\DeclareCaptionLabelSeparator{dot}{.~}
\usepackage{subcaption}
\usepackage{dsfont}
\usepackage{amsfonts}
\usepackage{graphicx,wrapfig}
\usepackage[T1]{fontenc}
\usepackage{amsmath}
\usepackage{mathtools, cuted}
\usepackage{amsthm}
\usepackage{amstext}
\usepackage{amssymb}
\usepackage{mathrsfs}
\usepackage{mathtools}
\usepackage{tikz}
\usepackage{marginnote}
\usepackage{tkz-tab}
\usetikzlibrary{topaths,calc}
\usepackage{pgfplots}
\usepackage{stackengine}
\usepackage{cite}

\usetikzlibrary{arrows,positioning, shapes}
\usepackage{float}

\usepackage{setspace}
\usepackage{hyperref}
\hypersetup{colorlinks,
   colorlinks,
    citecolor=red,
    filecolor=green,
  linkcolor=blue,
    linktoc=page,
   urlcolor=blue
}

\def\R{\mathbb{R}}

\def\eps{\varepsilon}

\def\A{{\mathcal A}}

\def\E{{\mathbb E}}
\def\F{{\mathcal F}}
\def\H{{\mathcal H}}
\def\I{{\mathcal I}}

\def\P{{\mathcal P}}

\def\T{{\mathcal T}}
\def\X{{\mathcal X}}
\def\Y{{\mathcal Y}}
\def\M{{\mathcal M}}
\def\N{{\mathcal N}}

\def\L{{\mathcal L}}

\def\U{{\mathcal U}}
\def\V{{\mathcal V}}

\def\sBer{{\mathsf{Bernoulli}}}

\def\sU{{\mathsf U}}
\def\sL{{\mathsf L}}

\def\sG{{\mathsf G}}

\def \var {{\mathsf {var}   }}
\def \mmse {{\mathsf {mmse}   }}
\def \cp {\mathsf{P}_{\mathsf{c}}}

\def \ib {\mathsf{IB}}
\def \pf {\mathsf{PF}}
\def \bsc {\mathsf{BSC}}
\def \kl {D_{\mathsf{KL}}}
\def \pf {\mathsf{PF}}
\def \pxy {P_{XY}}
\def \conv {\mathcal{K}_{\mathsf{\cup}}}

\def \conc {\mathcal{K}_{\mathsf{\cap}}}

\newcommand{\eq}[1]{\begin{equation*}
	#1
	\end{equation*}}
\newcommand{\eqn}[2]{\begin{equation}
	\label{#1}
	#2
	\end{equation}}

\usepackage[font=small,labelsep=space]{caption}
\DeclareCaptionLabelSeparator{dot}{.~}
\newcounter{example}
\newenvironment{example}[1][]{\refstepcounter{example}\par\medskip
   \noindent \textit{Example~\theexample. #1} \rmfamily}{\medskip}
\newtheorem{definition}{Definition}
\newtheorem{theorem}{Theorem}
\newtheorem{corollary}{Corollary}

\newtheorem{proposition}{Proposition}
\newtheorem{lemma}{Lemma}
\newtheorem{remark}{Remark}

\newcommand{\markov}{\mathrel\multimap\joinrel\mathrel-%
\mspace{-9mu}\joinrel\mathrel-}

\usepgfplotslibrary{fillbetween}
\usepackage{times}
\usepackage{tikz}
\usepackage{tkz-tab}
\usepackage{pgfplots}
\usepgfplotslibrary{fillbetween}
\usetikzlibrary{intersections}
\usepackage{amsmath}
\usepackage{verbatim}
\usetikzlibrary{arrows,shapes}
\tikzstyle{RectObject}=[rectangle,fill=white,draw,line width=0.2mm]
\tikzstyle{line}=[draw]
\tikzstyle{arrow}=[draw, -latex]
\usetikzlibrary{decorations.pathmorphing}
\usetikzlibrary{calc,shapes, positioning}
\usepackage{graphicx}
\usetikzlibrary{shapes.geometric}

\DeclareFontFamily{U}{BOONDOX-calo}{\skewchar\font=45 }
\DeclareFontShape{U}{BOONDOX-calo}{m}{n}{
	<-> s*[1.05] BOONDOX-r-calo}{}
\DeclareFontShape{U}{BOONDOX-calo}{b}{n}{
	<-> s*[1.05] BOONDOX-b-calo}{}
\DeclareMathAlphabet{\mathcalboondox}{U}{BOONDOX-calo}{m}{n}
\SetMathAlphabet{\mathcalboondox}{bold}{U}{BOONDOX-calo}{b}{n}
\DeclareMathAlphabet{\mathbcalboondox}{U}{BOONDOX-calo}{b}{n}

\definecolor{DukeBlue}{HTML}{001A57}
\definecolor{DarkRed}{rgb}{0.75, 0.0, 0.0}
\definecolor{DarkGreen}{rgb}{0.0, 0.5, 0.0}

\allowdisplaybreaks
\begin{document}

	\title{\vspace{5.5mm}   Bottleneck Problems:\\ Information and Estimation-Theoretic View\thanks{This work was supported in part by NSF under grants CIF 1922971, 1815361, 1742836, 1900750, and CIF CAREER 1845852.}}

\author{%
Shahab Asoodeh and Flavio P. Calmon\thanks{S. Asoodeh and F. P. Calmon are with School of Engineering and Applied Science, Harvard University (e-mails: \{shahab, flavio\}@seas.harvard.edu). Part of the results in this
paper was presented at the International Symposium on Information Theory 2018 \cite{Hsu_Generalizing}. }
}	
	\date{}
	\maketitle

	\begin{abstract}
		Information bottleneck (IB) and privacy funnel (PF) are two closely related optimization problems 
		which have found applications in machine learning, design of privacy algorithms, capacity problems (e.g., Mrs. Gerber's Lemma), strong data processing inequalities, among others.
		In this work, we first investigate the functional properties of 
		IB and PF through a unified theoretical framework. We then connect them to three information-theoretic coding problems, namely  hypothesis testing against independence, noisy source coding and dependence dilution. Leveraging these connections, we prove a new cardinality bound for the auxiliary variable in IB, making its computation more tractable for discrete random variables.    
		
		In the second part, we introduce a general family of optimization problems, termed as \textit{bottleneck problems}, by 	replacing mutual information in IB and PF with other notions of mutual information, namely  $f$-information and Arimoto's mutual information.  We then argue that, unlike IB and PF, these problems lead to easily interpretable guarantee in  a variety of inference tasks with statistical constraints on accuracy and privacy. Although the underlying optimization problems are non-convex, we develop a technique to evaluate bottleneck problems in closed form by equivalently expressing them in terms of lower convex or upper concave envelope of certain functions.  By applying this technique to binary case, we derive closed form expressions for several bottleneck problems. 
	\end{abstract}
	
\tableofcontents

\section{Introduction}
Optimization formulations that involve information-theoretic quantities (e.g., mutual information) have been instrumental in a variety of learning problems found in machine learning.
A notable example is the \textit{information bottleneck} ($\ib$) method \cite{tishby2000information}.
Suppose $Y$ is a target variable and $X$ is an observable correlated variable with joint distribution $\pxy$. The goal of $\ib$  is to learn  a "compact" summary (aka \textit{bottleneck}) $T$ of $X$ that is \textit{maximally} "informative" for inferring $Y$. 
The bottleneck variable $T$ is assumed to be  generated from $X$ by applying a \textit{random} function $F$ to $X$, i.e., $T= F(X)$, in such a way that it is conditionally independent of $Y$ given $X$, that we denote by
	\begin{equation}\label{Eq:Markov}
	    Y\markov X\markov T.
	\end{equation} 
The $\ib$  quantifies this goal by measuring  the ``compactness'' of $T$ using the mutual information $I(X;T)$ and, similarly,  ``informativeness'' by $I(Y;T)$. For a given level of compactness $R\geq 0$, $\ib$  extracts the bottleneck variable $T$ that solves the constrained optimization problem
\eqn{IB00}{\ib(R) \coloneqq \sup~ I(Y; T)\qquad \text{subject~ to} \qquad I(X; T)\leq R,}
where the supremum is taken over all randomized functions $T = F(X)$ satisfying $Y\markov X\markov T$.  

The optimization problem that underlies the information bottleneck has been studied in the information theory literature as early as the 1970's --- see \cite{Gerber,Witsenhausen_Wyner,Gerbator_Ahlswede1977, Wyner_Gerber} --- as a technique to prove impossibility results in information theory and also to study the common information between $X$ and $Y$. 
Wyner and Ziv \cite{Gerber} explicitly determined the value of $\ib(R)$ for the special case of binary $X$ and $Y$ --- a result widely known as \textit{Mrs. Gerber's Lemma} \cite{Gerber,networkinfotheory}. 
More than twenty years later, the information bottleneck function was studied by Tishby et al. \cite{tishby2000information} and re-formulated in a data analytic context.
Here, the random variable $X$ represents a high-dimensional observation with a corresponding  low-dimensional feature $Y$. $\ib$  aims at specifying a compressed description of image which is maximally informative about feature $Y$. This framework led to several applications in clustering \cite{slonim2000document,IB_clustering,Agglomerative_IB} and quantization \cite{Quantization_IB,Quantization2_IB}.

A closely-related framework to $\ib$  is the \textit{privacy funnel} ($\pf$)  problem \cite{Makhdoumi2014FromTI, Calmon_fundamental-Limit, Asoodeh_Allerton}. In the $\pf$ framework, a bottleneck variable $T$ is sought to maximally preserve "information" contained in $X$ while revealing as little about $Y$ as possible. This framework aims to capture the inherent trade-off between revealing $X$ perfectly and leaking a sensitive attribute $Y$. For instance, suppose a user wishes to share an image $X$ for some classification tasks.  The image might carry information about attributes, say $Y$, that the user might consider as
sensitive, even when such information is of limited use for the tasks, e.g,  location, or emotion. The $\pf$  framework seeks to extract a representation of $X$ from which the original image can be recovered with maximal accuracy while minimizing the privacy leakage with respect to $Y$.  Using mutual information for both privacy leakage and informativeness, the  privacy funnel can be formulated as    
 \eqn{PF00}{\pf(r) \coloneqq \inf~ I(Y; T)\qquad \text{subject~ to} \qquad I(X; T)\geq r,}
where the infumum is taken over all randomized function $T = F(X)$ and $r$ is the parameter specifying the level of informativeness.
It is evident from the formulations \eqref{IB00} and \eqref{PF00} that $\ib$  and $\pf$  are closely related. In fact, we shall see later that they correspond to the upper and lower boundaries of a two-dimensional compact convex set. This duality has led to design of 
greedy algorithms \cite{Makhdoumi2014FromTI, Sadeghi_PF} for estimating $\pf$  based on the \textit{agglomerative information bottleneck} \cite{Agglomerative_IB} algorithm.
A similar formulation has recently been proposed in \cite{PF_Adverserially} as a  tool to train a neural network for learning a private representation of data $X$. Solving $\ib$  and $\pf$  optimization problems analytically is challenging.  However, recent machine learning applications, and deep learning algorithms in particular, have reignited the study of both $\ib$  and $\pf$ (see Related Work).

In this paper, we first give a cohesive overview of the existing results surrounding the $\ib$ and the $\pf$ formulations. We then provide a comprehensive analysis of $\ib$  and $\pf$ from an information-theoretic perspective, as well as a survey of several formulations connected to the $\ib$ and $\pf$ that have been introduced in the information  theory and machine learning literature.   Moreover, we overview connections with 
coding problems such as remote
source-coding \cite{Noisy_SourceCoding_Dobrushin}, testing against independence \cite{Hypothesis_Testing_Ahslwede}, and dependence dilution \cite{Asoode_submitted}. Leveraging these
connections, we prove a new cardinality bound for the bottleneck variable in $\ib$, leading to more tractable optimization problem for $\ib$. We then consider a broad family of optimization problems by going beyond mutual information in formulations \eqref{IB00} and \eqref{PF00}. We propose two candidates for this task:  \textit{Arimoto's mutual information} \cite{Arimoto_Original_Paper} and $f$-information \cite{Maxim_Strong_TIT}. By replacing $I(Y; T)$ and/or $I(X;T)$ with either of these measures, we generate a family of optimization problems that we referred to as the \textit{bottleneck problems}.  These problems are shown to better capture the underlying trade-offs intended by $\ib$ and $\pf$. More specifically, our main contributions are listed next.
\begin{itemize}
\item Computing $\ib$  and $\pf$  are notoriously challenging when $X$ takes values in a set with infinite cardinality (e.g., $X$ is drawn from a continuous probability distribution). We consider three different scenarios to circumvent this difficulty. First, we assume that $X$ is a Gaussian perturbation of $Y$, i.e., $X = Y + N^\sG$ where $N^\sG$ is a noise variable sampled from a Gaussian distribution independent of $Y$. Building upon the recent advances in entropy power inequality in \cite{EPI_Courtade}, we derive a sharp upper bound for $\ib(R)$. As a special case, we consider jointly Gaussian $(X,Y)$ for which the upper bound becomes tight. This then provides a significantly simpler proof for the fact that in this special case the optimal bottleneck variable $T$ is also Gaussian than the original proof given in \cite{GaussianIB}. 
In the second scenario, we assume that $Y$ is a Gaussian perturbation of $X$, i.e., $Y = X + N^\sG$. This corresponds to a practical setup where the feature $Y$ might be perfectly obtained from a noisy observation of $X$. Relying on the recent results in strong data processing inequality \cite{calmon2015strong}, we obtain an upper bound on $\ib(R)$ which is tight for small values of $R$. In the last scenario, we compute second-order approximation of $\pf(r)$ under the assumption that $T$ is obtained by Gaussian perturbation of $X$, i.e., $T = X + N^\sG$. Interestingly, the rate of increase of $\pf(r)$  for small values of $r$ is shown to be dictated by an asymmetric measure of dependence introduced by R\'enyi \cite{Renyi-dependence-measure}.

\item We extend the Witsenhausen and Wyner's approach \cite{Witsenhausen_Wyner} for analytically computing $\ib$ and $\pf$. This technique  converts solving the optimization problems in $\ib$ and $\pf$ to determining the convex and concave envelopes of a certain function, respectively. We apply this technique to binary $X$ and $Y$ and derive a closed form expression for $\pf(r)$-- we call this result \textit{Mr. Gerber's Lemma}.

\item Relying on the connection between $\ib$ and noisy source coding \cite{Noisy_SourceCoding_Dobrushin} (see \cite{Bottleneck_Polyanskiy, Bottleneck_Shamai}), we show that the optimal bottleneck variable $T$ in optimization problem \eqref{IB00} takes values in a set $\T$ with cardinality $|\T|\leq |\X|$. Compared to the best cardinality bound previously known (i.e., $|\T|\leq |\X|+1$), this result leads to a  reduction in the search space's dimension of the optimization problem \eqref{IB00} from $\R^{|\X|^2}$ to $\R^{|\X|(|\X|-1)}$. Moreover, we show that this does not hold for $\pf$, indicating a fundamental difference in optimizations problems \eqref{IB00} and \eqref{PF00}.

\item Following  \cite{strouse2017dib, Asoodeh_Allerton}, we study the \textit{deterministic} $\ib$ and $\pf$ (denoted by $\mathsf{dIB}$ and $\mathsf{dPF}$) in which $T$ is assumed to be a deterministic function of $X$, i.e., $T = f(X)$ for some function $f$. By connecting $\mathsf{dIB}$ and $\mathsf{dPF}$ with  entropy-constrained scalar quantization problems in information theory \cite{Polyanskiy_Distilling}, we obtain bounds on them explicitly in terms of $|\X|$. Applying these bounds to $\ib$, we obtain that 
$\frac{\ib(R)}{I(X;Y)}$ is bounded by one from above and by $\min\{\frac{R}{H(X)}, \frac{e^R-1}{|\X|}\}$ from below. 
\item By replacing $I(Y; T)$ and/or $I(X; T)$ in \eqref{IB00} and \eqref{PF00} with Arimoto's mutual information or $f$-information, we generate a family of bottleneck problems. 
We then argue that these new functionals better describe the trade-offs that were intended to be captured by $\ib$ and $\pf$. 
The main reason is three-fold: First, as illustrated in Section~\ref{sec:Operational}, mutual information in $\ib$ and $\pf$ are mainly justified when  $n\gg 1$ independent samples $(X_1, Y_1), \dots, (X_n, Y_n)$ of $\pxy$ are considered. However, Arimoto's mutual information allows for operational interpretation even in the single-shot regime (i.e., for $n=1$). Second, $I(Y; T)$ in $\ib$ and $\pf$ is meant to be a proxy for the efficiency of reconstructing $Y$ given observation $T$. However, this can be accurately formalized by probability of correctly guessing $Y$ given $T$ (i.e., Bayes risk) or minimum mean-square error (MMSE) in estimating $Y$ given $T$. While $I(Y; T)$ bounds these two measures, we show that they are \textit{precisely} characterized by Arimoto's mutual information and $f$-information, respectively. 
Finally, when $\pxy$ is unknown, mutual information is known to be  notoriously difficult to estimate. Nevertheless, Arimoto's mutual information and $f$-information are easier to estimate: While mutual information can be estimated with estimation error that scales as $O(\log n/\sqrt{n})$ \cite{Shamir_IB}, Diaz et a. \cite{Diaz_Robustness} showed that this estimation error for Arimoto's mutual information and $f$-information is $O(1/\sqrt{n})$.

We also generalize our computation  technique that enables us to analytically compute these bottleneck problems. Similar as before, this technique converts computing bottleneck problems to determining convex and concave envelopes of certain functions. Focusing on binary $X$ and $Y$, we derive closed form expressions for some of the bottleneck problems.      
\end{itemize}

\subsection{Related Work}
The $\ib$ formulation has been extensively applied in representation learning and clustering \cite{IB_clustering, IB_DocumentClustering, IB_DoubleClustering, IB_Hidden,Zaidi_distributedIB,Zaidi2019distributed}. Clustering based on $\ib$ results in algorithms that cluster data points in terms of the similarity of $P_{Y|X}$. When data points lie in a metric space, usually geometric clustering is preferred where clustering is based upon the geometric  (e.g., Euclidean) distance.   
Strouse and Schwab \cite{strouse2017dib, strouse2019clustering} proposed the \textit{deterministic} $\ib$ (denoted by $\mathsf{dIB}$)  by enforcing that  $P_{T|X}$ is a deterministic mapping: 
$\mathsf{dIB}(R)$ denotes the supremum of $I(Y; f(X))$ over all functions $f:\X\to \T$ satisfying $H(f(X))\leq R$. This optimization problem is closely related to the problem of scalar quantization in information theory:  designing a function $f:\X\to [M]\coloneqq \{1, \dots, M\}$ with a pre-determined output alphabet with $f$ optimizing some objective functions. This objective might be maximizing or minimizing $H(f(X))$  \cite{Cicalese} or maximizing $I(Y; f(X))$ for a random variable $Y$ correlated with $X$ \cite{Polyanskiy_Distilling, Lapidoth_Koch, LDPC1_quantization, LDPC2_quantization}. Since $H(f(X))\leq \log M$ for $f:\X\to [M]$,  the latter problem provides lower bounds for $\mathsf{dIB}$ (and thus for $\ib$). In particular, one can exploit \cite[Theorem 1]{LDPC3_quantization}  to obtain
$I(X; Y) - \mathsf{dIB}(R)\leq O(e^{-2 R/|\Y|-1})$ provided that $\min\{|\X|, 2^R\} >2 |\Y|$. This result establishes a \textit{linear} gap between $\mathsf{dIB}$ and $I(X; Y)$ irrespective of $|\X|$. 

The connection between quantization and $\mathsf{dIB}$ further allows us to obtain \textit{multiplicative} bounds. For instance, if $Y\sim \sBer(\frac{1}{2})$ and $X = Y + N^\sG$, where $N^\sG\sim \N(0, 1)$ is independent of $Y$, then it is well-known in information theory literature that $I(Y; f(X))\geq \frac{2}{\pi}I(X; Y)$ for all non-constant $f:\X\to \{0,1\}$ (see, e.g., \cite[Section 2.11]{Viterbi}), thus $\mathsf{dIB}(R)\geq \frac{2}{\pi}I(X; Y)$ for  $R\leq 1$. We further explore this connection to provide  multiplicative bounds on $\mathsf{dIB}(R)$ in Section~\ref{Sec:DIB}. 

The study of $\ib$  has recently gained increasing traction  in the context of deep learning. By taking $T$ to be the activity of the hidden layer(s), Tishby and Zaslavsky \cite{tishby2015deep} (see also \cite{IB_DP_openBox}) argued that neural network classifiers trained with cross-entropy loss and stochastic gradient descent (SGD) inherently aims at solving the $\ib$  optimization problems. In fact,  it is claimed that the graph of the function $R\mapsto \ib(R)$ (the so-called the information plane) characterizes the  learning dynamic of different layers in the network:   shallow layers correspond to maximizing $I(Y;T)$ while deep layers'  objective is minimizing $I(X;T)$. While the generality of this claim was refuted empirically  in \cite{On_IB_DL} and theoretically in \cite{Inf_flow_IB_Polyiansky,Amjad_IB}, it inspired significant follow-up studies. These include (i) modifying  neural network training in order to solve the $\ib$  optimization problem \cite{alemi2016deep, kolchinsky2017nonlinear, kolchinsky2018caveats,ReleventSparseCode, wickstrom2020information}; (ii) creating connections between $\ib$  and generalization error \cite{Piantanida_roleIB}, robustness \cite{alemi2016deep}, and detection of out-of-distribution data \cite{Alemi_Uncertainity}; and (iii) using $\ib$  to understand specific characteristic of neural networks  \cite{Yu2018UnderstandingCN, Cheng2018EvaluatingCO, wickstrom2020information,Higgins2017betaVAELB}. 

In both $\ib$ and $\pf$, mutual information poses some limitations. For instance, it may become infinity in deterministic neural networks \cite{On_IB_DL, Inf_flow_IB_Polyiansky,Amjad_IB} and also may not lead to proper privacy guarantee \cite{Issa_Leakage_TIT}. 
As suggested in \cite{wickstrom2020information, Sufficient_Statistics}, one way to address this issue is to replace mutual information with other statistical measures.  In the privacy literature, several measures with strong privacy guarantee have been proposed including  
R\'enyi maximal correlation \cite{Asoodeh_CWIT, Asoode_submitted, Fawaz_Makhdoumi}, probability of correctly recovering \cite{Asoodeh_TIT19, Asoode_ISIT17}, minimum mean-squared estimation error (MMSE) \cite{Asoode_MMSE_submitted, Calmon_principal_TIT}, $\chi^2$-information \cite{Hao_Privacy_estimation} (a special case of $f$-information to be described in Section~\ref{Sec:Family}), Arimoto's and Sibson's mutual information \cite{Shahab_PhD_thesis, Issa_Leakage_TIT} -- to be discussed in Section~\ref{Sec:Family}, maximal leakage \cite{Liao_maximal_leakage}, and local differential privacy \cite{privacyaware}. All these measures ensure \textit{interpretable} privacy guarantees. For instance, it is shown in\footnote{The original results in \cite{Asoode_MMSE_submitted, Calmon_principal_TIT} involve R\'enyi maximal correlation instead of $\chi^2$-information. However, it can be shown that $\chi^2$-information is equal to the sum of squares of the singular values of $f(Y)\mapsto \E[f(Y)|T]$ minus one (the largest one), while R\'enyi maximal correlation is equal to the second largest singular value \cite{Witsenhausen:dependent}. Thus, $\chi^2$-information upper bounds R\'enyi maximal correlation.} \cite{Asoode_MMSE_submitted, Calmon_principal_TIT} that if $\chi^2$-information between $Y$ and $T$ is sufficiently small, then no \textit{functions} of $Y$ can be efficiently reconstructed given $T$; thus providing an interpretable privacy guarantee. 

Another limitation of mutual information is related to its estimation difficulty. It is known that mutual information can be estimated from $n$ samples with the estimation error that scales as $O(\log n/\sqrt{n})$ \cite{Shamir_IB}. However, as shown by Diaz et al.~\cite{Diaz_Robustness}, the estimation error for most of the above measures scales as $O(1/\sqrt{n})$. Furthermore, the recently popular variational estimators for mutual information, typically implemented via deep learning methods \cite{MI_Estimator_Poole, MINE_Belghazi, Contrastive}, presents some fundamental limitations \cite{Understanding_Variational}: the variance of the estimator might grow \textit{exponentially} with the ground truth mutual information and also the estimator might  not satisfy basic properties of mutual information such as data processing inequality or additivity. McAllester and Stratos \cite{Stratos_MI_Estimator} showed that some of these limitations are inherent to a large family of mutual information estimators.

\subsection{Notation}
We use capital letters, e.g., $X$, for random variables and calligraphic letters for their alphabets, e.g.,  $\X$. If $X$ is distributed according to probability mass function (pmf) $P_X$, we write $X\sim P_X$. Given two random variables $X$ and $Y$, we write $P_{XY}$ and $P_{Y|X}$ as the joint distribution and the conditional distribution of $Y$ given $X$. We also interchangeably refer to $P_{Y|X}$ as a \textit{channel} from $X$ to $Y$. 
We use $H(X)$ to denote both entropy and differential entropy of $X$, i.e., we have 
$$H(X) = -\sum_{x\in \X}P_X(x)\log P_X(x)$$
if $X$ is a discrete random variable taking values in $\X$ with probability mass function (pmf) $P_X$ and 
$$H(X) = -\int\log f_X(x) \log f_X(x) \text{d}x,$$
where $X$ is an absolutely continuous random variable with probability density function (pdf) $f_X$. 
If $X$ is a binary random variable with $P_X(1) = p$, we write $X\sim \sBer(p)$. In this case, its entropy is called \textit{binary entropy function} and denoted by $h_\mathsf{b}(p) \coloneqq -p\log p - (1-p)\log(1-p)$. 
We use superscript $\sG$ to describe a standard Gaussian random variable, i.e., $N^\sG\sim \N(0, 1)$.
Given two random variables $X$ and $Y$, their (Shannon's) mutual information is denoted by $I(X; Y)\coloneqq H(Y) - H(Y|X)$. We let $\P(\X)$ denote the set of all probability distributions on the set $\X$.  Given an arbitrary $Q_X\in \P(\X)$ and a channel $P_{Y|X}$, we let $Q_X P_{Y|X}$ denote the resulting output distribution on $\Y$. 
For any $a\in [0,1]$, we use $\bar a$ to denote $1-a$ and for any integer $k\in \mathbb N$, $[k]\coloneqq \{1, 2, \dots, k\}$. 

Throughout the paper, we assume a pair of (discrete or continuous) random variables $(X,Y)\sim \pxy$ are given with a fixed joint distribution $\pxy$, marginals $P_X$ and $P_Y$, and conditional distribution $P_{Y|X}$. We then use $Q_X\in \P(\X)$ to denote an \textit{arbitrary} distribution with $Q_Y = Q_X P_{Y|X}\in \P(\Y)$.

\section{Information Bottleneck and Privacy Funnel:  Definitions and Functional Properties} \label{Sec:IB_PF}
	In this section, we review the  \textit{information bottleneck}  and its closely related functional, the \textit{privacy funnel}. We then prove some analytical properties of these two functionals and develop a convex analytic approach which enables us to compute closed-form expressions for both these two functionals in some simple cases.   
	
	To precisely quantify the trade-off between these two conflicting goals, the $\ib$ optimization problem \eqref{IB00} was proposed \cite{tishby2000information}. Since any randomized function $T = F(X)$ can be equivalently characterized by a conditional distribution, \eqref{IB00} can be instead expressed as 
	\begin{equation}\label{eq:Def_IB}
	\mathsf{IB}(P_{XY}, R)\coloneqq \sup_{\substack{P_{T|X}:Y\markov X\markov T \\ I(X;T)\leq R}} I(Y; T), \qquad \text{or}\qquad  \widetilde{\mathsf{IB}}(P_{XY}, \tilde R)\coloneqq \inf_{\substack{P_{T|X}:Y\markov X\markov T \\ I(Y;T)\geq \tilde R}} I(X; T).
	\end{equation}
	where $R$ and $\tilde{R}$ denote the level of desired compression and informativeness, respectively. 
	We use $\ib(R)$ and $\widetilde\ib(\tilde R)$ to denote $\ib(\pxy, R)$ and $\widetilde\ib(\pxy, \tilde R)$, respectively, when the joint distribution is clear from the context. 
	Notice that if $\ib(\pxy, R) = \tilde R$,  then $\widetilde\ib(\pxy, \tilde R) = R$.  
	
	Now consider the setup where data $X$
	is required to be disclosed  while maintaining the privacy of a sensitive attribute, represented by $Y$. This goal was formulated by $\pf$ in \eqref{PF00}. As before, replacing randomized function $T = F(X)$ with conditional distribution $P_{T|X}$, we can equivalently express \eqref{PF00} as 
\begin{equation}\label{eq:Def_PF}
	\pf(P_{XY}, r)\coloneqq \inf_{\substack{P_{T|X}:Y\markov X\markov T \\ I(X;T)\geq r}} I(Y; T), \qquad \text{or} \qquad \widetilde\pf(P_{XY}, \tilde r)\coloneqq \sup_{\substack{P_{T|X}:Y\markov X\markov T \\ I(Y;T)\leq \tilde r}} I(X; T),
	\end{equation}
where $\tilde r$ and $r$ denote the level of desired privacy and informativeness, respectively. The case  $\tilde r = 0$ is particularly interesting in practice and specifies \textit{perfect privacy}, see e.g.,  \cite{Calmon_fundamental-Limit,Rassouli_Perfect}. As before, we write $\widetilde\pf( \tilde r)$ and $\pf(r)$ for $\widetilde\pf( \pxy, \tilde r)$ and $\pf(\pxy, r)$ when $\pxy$ is clear from the context.
\begin{figure}[t]
	\centering
	\begin{subfigure}[t]{0.32\textwidth}
		\includegraphics[height=4cm, width=1\columnwidth]{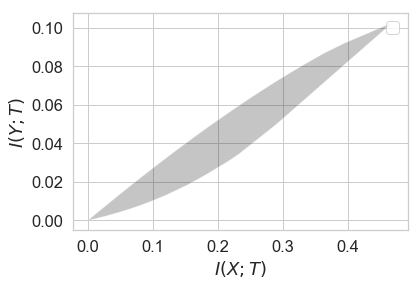}
		\caption{$|\X| = |\Y|=2$}
		\label{fig:perfect_Privacy}
		\end{subfigure}
		\begin{subfigure}[t]{0.32\textwidth}
		\includegraphics[height=4cm, width=1\columnwidth]{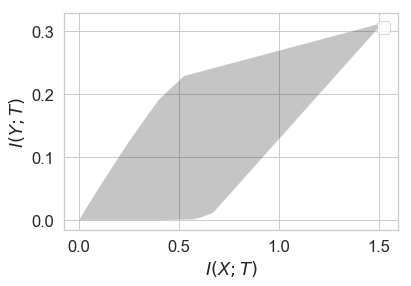}
		\caption{$|\X| = 3, |\Y|=2$}
		\end{subfigure}
		\begin{subfigure}[t]{0.32\textwidth}
		\includegraphics[height=4cm, width=1\columnwidth]{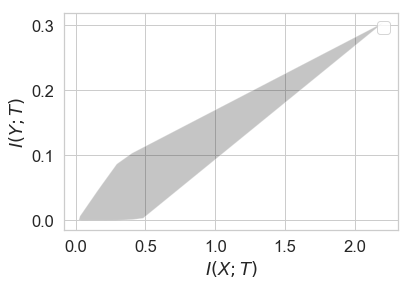}
		\caption{$|\X| = 5, |\Y|=2$}
		\end{subfigure}
		\caption{Examples of the set $\M$, defined in \eqref{Def:Set_M}. The upper and lower boundaries of this set correspond to $\ib$ and $\pf$, respectively.  It is worth noting that, while $\ib(R) =0$  only at $R=0$, $\pf(r) = 0$ holds in general for $r$ belonging to a non-trivial interval (only for $|\X|>2$). Also, note that in general neither upper nor lower boundaries are smooth. A sufficient condition for smoothness is $P_{Y|X}(y|x)>0$ (see Theorem~\ref{Thm:IB_Properties}), thus  both $\ib$ and $\pf$ are smooth in the binary case. }   
\label{fig:M_Set}
\end{figure}

The following properties of $\ib$ and $\pf$ follow directly from their definitions. The proof of this result (and any other results in this section) is given in Appendix~\ref{Appendix_ProofSecIB_PF}.
	\begin{theorem}\label{Thm:IB_Properties}
		For a given $P_{XY}$, the mappings $\ib(R)$ and $\pf(r)$ have the following properties:
		\begin{itemize}
			\item $\ib(0)=\pf(0)=0$.
			
			\item $\ib(R)= I(X;Y)$ for any $R\geq H(X)$ and $\pf(r) = I(X; Y)$ for $r\geq H(X)$.
			\item $0\leq\ib(R)\leq \min\{R, I(X;Y)\}$ for any $R\geq 0$ and $\pf(r)\geq \max\{r - H(X|Y), 0\}$ for any $r\geq 0$.
			\item $R\mapsto \ib(R)$ is continuous, strictly increasing,  and concave on the range $(0, I(X;Y))$.
			\item $r\mapsto \pf(r)$ is continuous, strictly increasing, and convex on the range $(0, I(X;Y))$.
			\item If $P_{Y|X}(y|x)>0$ for all $x\in \X$ and $y\in \Y$, then both $R\mapsto \ib(R)$ and $r\mapsto \pf(r)$ are continuously differentiable over $(0, H(X))$.
			\item $R\mapsto \frac{\ib(R)}{R}$ is non-increasing and $r\mapsto \frac{\pf(r)}{r}$ is non-decreasing.
			\item We have $$\ib(R)\coloneqq \sup_{\substack{P_{T|X}:Y\markov X\markov T \\ I(X;T)= R}} I(Y; T), \qquad \text{and}\qquad \pf(r)\coloneqq \inf_{\substack{P_{T|X}:Y\markov X\markov T \\ I(X;T)= r}} I(Y; T).$$
		\end{itemize}
	\end{theorem}   
	
	 According to this theorem, we can always restrict both $R$ and $r$ in \eqref{eq:Def_IB} and \eqref{eq:Def_PF}, respectively, to $[0, H(X)]$ as  $\ib(R) = \pf(r) = I(X;Y)$ for all $r, R\geq H(X)$.

Define $\M = \M(\pxy)\subset \R^2$ as
\begin{equation}\label{Def:Set_M}
    \M \coloneqq \big\{(I(X; T), I(Y; T)):Y\markov X\markov T, (X, Y)\sim \pxy\big\}.
\end{equation}
It can be directly verified that $\M$ is convex.  According to this theorem, $R\mapsto \ib(R)$ and $r\mapsto \pf(r)$ correspond to the upper and lower boundary of $\M$, respectively. The convexity of $\M$ then implies the concavity and convexity of $\ib$ and $\pf$.  Fig.~\ref{fig:M_Set} illustrates the set $\M$ for the simple case of binary $X$ and $Y$.

	 While  both  $\ib(0) = 0$ and  $\pf(0) = 0$, their behavior in the neighborhood around zero might be completely different. As illustrated in Fig.~\ref{fig:M_Set}, $\ib(R)>0$ for all $R>0$, whereas $\pf(r) = 0$ for $r\in [0, r_0]$ for some $r_0>0$. When such $r_0>0$ exists, we say  perfect privacy occurs: there exists a variable $T$ satisfying $Y\markov X\markov T$ such that $I(Y; T) = 0$ while $I(X; T)>0$; making $T$ a representation of $X$ having perfect privacy (i.e., no information leakage about $Y$).  A necessary and sufficient condition for the existence of such $T$ is given in \cite[Lemma 10]{Asoode_submitted} and \cite[Theorem 3]{Calmon_fundamental-Limit}, described next.
	 \begin{theorem}[Perfect privacy]\label{Thm:perfect_Privacy}
	 Let $(X, Y)\sim \pxy$ be given and   $\A\subset [0, 1]^{|\Y|}$ be the set of vectors $\{P_{Y|X}(\cdot|x), x\in \X\}$. Then there exists $r_0>0$ such that $\pf(r) = 0$ for $r\in [0, r_0]$ if and only if vectors in $\A$ are linearly independent. 
	 \end{theorem}
	 In light of this theorem, we obtain  that perfect privacy occurs if $|\X|>|\Y|$. It also follows from the theorem that for binary $X$, perfect privacy cannot occur (see Fig.~\ref{fig:perfect_Privacy}). 
	 
	Theorem~\ref{Thm:IB_Properties} enables us to derive a simple bounds for $\ib$ and $\pf$. Specifically, the facts that $\frac{\pf(r)}{r}$ is non-decreasing and $\frac{\ib(R)}{R}$ is non-increasing immediately result in the the following linear bounds.
	 \begin{theorem}[Linear lower bound] \label{Thm:Bounds_IB_PF}
	 For $r, R\in (0, H(X))$, we have 
	 \begin{equation}
	     \inf_{Q_X\in \P(\X)\atop Q_X\neq P_X}\frac{\kl(Q_Y\|P_Y)}{\kl(Q_X\|P_X)}\leq \frac{\pf(r)}{r}\leq \frac{I(X;Y)}{H(X)}\leq \frac{\ib(R)}{R} \leq \sup_{Q_X\in \P(\X)\atop Q_X\neq P_X}\frac{\kl(Q_Y\|P_Y)}{\kl(Q_X\|P_X)} \leq 1.
	 \end{equation}
	 \end{theorem}

	 In light of this theorem, if $\pf(r) = r$, then $I(X; Y) = H(X)$, implying $X = g(Y)$ for a deterministic function $g$. Conversely, if $X = g(Y)$ then $\pf(r) = r$ because for all $T$ forming the Markov relation $Y\markov g(Y)\markov T$, we have $I(Y; T) = I(g(Y); T)$.
	 On the other hand, we have $\ib(R) = R$ if and only if there exists a variable $T^*$ satisfying $I(X; T^*) = I(Y; T^*)$ and thus the following double Markov relations
	 $$Y\markov X\markov T^*,\qquad \text{and}\qquad X\markov Y\markov T^*.$$
	 It can be verified (see \cite[Problem 16.25]{csiszarbook}) that this double Markov condition is equivalent to the existence of a pair of functions $f$ and $g$ such that $f(X) = g(Y)$ and $(X, Y)\markov f(X) \markov T^*$.
	 One special case of this setting, namely where $g$ is an identity function, has been recently studied in details in \cite{kolchinsky2018caveats} and will be reviewed in Section~\ref{Sec:DIB}.
	 Theorem~\ref{Thm:Bounds_IB_PF} also enables us to characterize the "worst" joint distribution $\pxy$ with respect to $\ib$ and $\pf$. As demonstrated in the following lemma, if $P_{Y|X}$ is an \textit{erasure} channel then $\frac{\pf(r)}{r} = \frac{\ib(R)}{R} = \frac{I(X; Y)}{H(X)}$.  
\begin{lemma}\label{lem:pf_BEC}
\begin{itemize}
    \item Let $\pxy$ be such that $\Y = \X\cup \{\perp\}$, $P_{Y|X}(x|x) = 1-\delta$, and $P_{Y|X}(\perp | x) = \delta$ for some $\delta>0$. Then
$$\frac{\pf(r)}{r} = \frac{\ib(R)}{R} = 1-\delta.$$
    \item Let $\pxy$ be such that $\X = \Y\cup \{\perp\}$, $P_{X|Y}(y|y) = 1-\delta$, and $P_{X|Y}(\perp | y) = \delta$ for some $\delta>0$. Then
	     $$\pf(r) = \max\{r - H(X|Y), 0\}.$$
\end{itemize}
\end{lemma}
The bounds in Theorem~\ref{Thm:Bounds_IB_PF} hold for all $r$ and $R$ in the interval $[0, H(X)]$. We can, however, improve them when $r$ and $R$ are sufficiently small. Let $\pf'(0)$ and $\ib'(0)$ denote the slope of $\pf(\cdot)$ and $\ib(\cdot)$ at zero, i.e., $\pf'(0) \coloneqq \lim_{r\to 0^+}\frac{\pf(r)}{r}$ and $\ib'(0) \coloneqq \lim_{R\to 0^+}\frac{\ib(R)}{R}$.  
\begin{theorem}\label{Thm:Derivative_IB_PF}
    Given $(X, Y)\sim \pxy$, we have 
    \begin{align*}
    \inf_{Q_X\in \P(\X)\atop Q_X\neq P_X}\frac{\kl(Q_Y\|P_Y)}{\kl(Q_X\|P_X)}= \pf'(0)&\leq \min_{\substack{x\in \X:\\P_{X}(x)>0}}\frac{\kl(P_{Y|X}(\cdot|x)\|P_Y(\cdot))}{-\log P_X(x)} \\
    & \leq \max_{\substack{x\in \X:\\P_{X}(x)>0} }\frac{\kl(P_{Y|X}(\cdot|x)\|P_Y(\cdot))}{-\log P_X(x)} \leq \ib'(0) = \sup_{Q_X\in \P(\X)\atop Q_X\neq P_X}\frac{\kl(Q_Y\|P_Y)}{\kl(Q_X\|P_X)}.
    \end{align*}
\end{theorem} 

This theorem provides the exact values of $\pf'(0)$ and $\ib'(0)$ and also simple bounds for them.  Although the exact expressions for $\pf'(0)$ and $\ib'(0)$ are usually difficult to compute,  a simple plug-in estimator is proposed in \cite{Hypercontractivity_NIPS2017}  for $\ib'(0)$. This estimator can be readily adapted to estimate $\pf'(0)$. 
Theorem~\ref{Thm:Derivative_IB_PF} reveals a profound connection between $\ib$ and the strong data processing inequality (SDPI) \cite{Ahlswede_Gacs}. More precisely, thanks to the pioneering work of Anantharam et al. \cite{anantharam}, it is known that the supremum of $\frac{\kl(Q_Y\|P_Y)}{\kl(Q_X\|P_X)}$  over all $Q_X\neq P_X$ is equal the supremum of $\frac{I(Y; T)}{I(X; T)}$ over all $P_{T|X}$ satisfying $Y\markov X\markov T$ and hence $\ib'(0)$ specifies the strengthening of the data processing inequality of mutual information. This connection may open a new avenue for new theoretical results for $\ib$, especially when $X$ or $Y$ are continuous random variables. In particular, the recent \textit{non-multiplicative} SDPI results \cite{Polyanskiy, calmon2015strong} seem insightful for  this purpose.

In many practical cases, we might have $n$ i.i.d. samples  $(X_1, Y_1), \dots, (X_n, Y_n)$ of $(X,Y)\sim P_{XY}$. We now study how $\mathsf{IB}$ behaves in $n$.  
   Let $X^n\coloneqq (X_1, \dots, X_n)$ and $Y^n\coloneqq (Y_1, \dots, Y_n)$.  
   Due to the i.i.d.\ assumption, we have   $P_{X^nY^n}(x^n,y^n)=\prod_{i=1}^nP_{XY}(x_i, y_i)$. This can also be described by independently feeding $X_i$, $i\in [n]$, to channel $P_{Y|X}$ producing $Y_i$. The following theorem, demonstrated first in \cite[Theorem 2.4]{Witsenhausen_Wyner}, gives a formula for $\ib$ in terms of $n$. 
  \begin{theorem}[Additivity]\label{Thm:IB_Additivity}
  We have 
  $$\frac{1}{n}\mathsf{IB}(P_{X^nY^n}, nR)=\mathsf{IB}(P_{XY}, R).$$ 	
   \end{theorem}
This theorem demonstrates that an optimal channel $P_{T^n|X^n}$ for i.i.d.\ samples $(X^n, Y^n)\sim \pxy$ is obtained by the Kronecker product of an optimal channel $P_{T|X}$ for $(X, Y)\sim \pxy$. This, however, may not hold in general for $\pf$, that is, we might have $\pf(P_{X^nY^n}, nr)< n \pf(\pxy, r)$, see \cite[Proposition 1]{Calmon_fundamental-Limit} for an example. 

\subsection{Gaussian $\ib$ and $\pf$}
In this section, we turn our attention to a special, yet important, case where $X = Y + \sigma N^\mathsf{G}$, where $\sigma>0$  and $N^\mathsf{G}\sim \N(0, 1)$ is independent of $Y$.
This setting subsumes the popular case of jointly Gaussian $(X, Y)$ whose information bottleneck functional was computed in \cite{Bottleneck_Gaussian} for the vector case (i.e., $(X,Y)$ are jointly Gaussian random vectors).  
\begin{lemma}\label{Lem:IB_UB_Courtade}
    Let $\{Y_i\}_{i=1}^n$ be $n$ i.i.d.\ copies of $Y\sim P_{Y}$ and $X_i = Y_i + \sigma N_i^\mathsf{G}$ where $\{N^\mathsf{G}_i\}$ are i.i.d samples of $\N(0, 1)$ independent of $Y$. Then, we have 
    $$\frac{1}{n}\ib(P_{X^nY^n}, nR)\leq H(X) - \frac{1}{2}\log\left[2\pi e\sigma^2 + e^{2(H(Y)-R)}\right].$$
\end{lemma}
It is worth noting that this result was concurrently proved in \cite{Zaidi_GaussianCase}.
The main technical tool in the proof of this lemma is a strong version of the entropy power inequality \cite[Theorem 2]{EPI_Courtade} which holds even if $X_i$, $Y_i$, and $N_i$ are random vectors (as opposed to scalar). Thus, one can readily generalize Lemma~\ref{Lem:IB_UB_Courtade} to the vector case. Note that the upper bound established in this lemma holds \textit{without any assumptions on $P_{T|X}$}. This upper bound provides a significantly simpler proof for the well-known fact that for the jointly Gaussian $(X, Y)$, the optimal channel $P_{T|X}$ is Gaussian. This result was first proved in \cite{GaussianIB} and used in \cite{Bottleneck_Gaussian} to compute an expression of $\ib$ for the Gaussian case. 

\begin{corollary}\label{Cor:Gaussian_IB}
If $(X, Y)$ are jointly Gaussian with correlation coefficient $\rho$, then we have 
\begin{equation}\label{eq:IB_Gaussian}
    \ib(R) = \frac{1}{2}\log\frac{1}{1-\rho^2+\rho^2 e^{-2R}}.
\end{equation}
Moreover, the optimal channel $P_{T|X}$ is given by $P_{T|X}(\cdot|x) = \N(0, \tilde\sigma^2)$ for $\tilde\sigma^2 = \sigma_Y^2\frac{e^{-2R}}{\rho^2(1-e^{-2R})}$ where $\sigma_Y^2$ is the variance of $Y$.
\end{corollary}

In Lemma~\ref{Lem:IB_UB_Courtade}, we assumed that $X$ is a Gaussian perturbation of $Y$. However, in some practical scenarios, we might have $Y$ as a Gaussian perturbation of $X$. For instance, let $X$ represent an image and $Y$ be a feature of the image that can be perfectly obtained from a noisy observation of $X$. Then, the goal is to compress the image with a given compression rate while retaining maximal information about the feature. The following lemma,  which is an immediate consequence of \cite[Theorem 1]{calmon2015strong}, gives an upper bound for $\ib$ in this case.
\begin{lemma}\label{lemma:IB_Non_Linear_SDPI}
Let $X^n$ be $n$ i.i.d. copies of a random variable $X$ satisfying $\E[X^2] \leq 1$ and $Y_i$ be the result of passing $X_i$, $i\in [n]$, through a Gaussian channel $Y = X + \sigma N^\mathsf{G}$, where $\sigma> 0$ and $N^\mathsf{G}\sim \N(0, 1)$ is independent of $X$. Then, we have
\begin{equation}\label{Eq:IB_UB}
    \frac{1}{n}\ib(P_{X^nY^n}, nR) \leq  R-\Psi(R, \sigma),
\end{equation}
where 
\begin{equation}\label{Eq:PSI}
    \Psi(R, \sigma)\coloneqq \max_{x\in [0, \frac{1}{2}]}2\mathsf{Q}\left(\sqrt{\frac{1}{x\sigma^2}}\right)\left(R-h_\mathsf{b}(x)-\frac{x}{2}\log\left(1+\frac{1}{x\sigma^2}\right)\right),
\end{equation}
$\mathsf{Q}(t)\coloneqq \int_{t}^\infty\frac{1}{\sqrt{2\pi}}e^{-\frac{t^2}{2}}\textnormal{d}t$ is the Gaussian complimentary CDF and  $h_\mathsf{b}(a)\coloneqq -a\log(a)-(1-a)\log(1-a)$ for $a\in (0,1)$ is the binary entropy function.  
Moreover, we have   
\begin{equation}\label{eq:IB_Order}
    \frac{1}{n}\ib(P_{X^nY^n}, nR) \leq   R - e^{-\frac{1}{R\sigma^2}\log\frac{1}{R} + \Theta\left(\log\frac{1}{R}\right)}.
\end{equation}
\end{lemma}
Note that that  Lemma~\ref{lemma:IB_Non_Linear_SDPI} holds for any arbitrary $X$ (provided that $\E[X^2]\leq 1$) and hence \eqref{Eq:IB_UB} bounds information bottleneck functionals for a wide family of $P_{XY}$. 
However, the bound is loose in general for large values of $R$. For instance, if $(X, Y)$ are jointly Gaussian (implying $Y = X + \sigma N^\mathsf{G}$ for some $\sigma>0$), then the right-hand side of \eqref{Eq:IB_UB} does not reduce to \eqref{eq:IB_Gaussian}.
To show this, we numerically compute the upper bound \eqref{Eq:IB_UB} and compare it with the Gaussian information bottleneck \eqref{eq:IB_Gaussian} in Fig.~\ref{Fig:IBGaussian}.
\begin{figure}
    \centering
    \includegraphics[height =5cm, width=7cm]{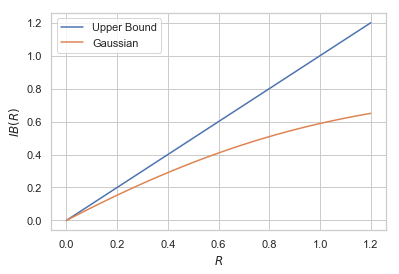}
    \caption{Comparison of \eqref{eq:IB_Gaussian}, the exact value of $\ib$ for jointly Gaussian $X$ and $Y$ (i.e., $Y = X + \sigma N^\mathsf{G}$ with $X$ and $N^\mathsf{G}$ being both standard Gaussian $\N(0, 1)$), with the general upper bound \eqref{Eq:IB_UB} for $\sigma^2 = 0.5$. It is worth noting that while the Gaussian $\ib$ converges to $I(X; Y)\approx 0.8$, the upper bound diverges.   }
    \label{Fig:IBGaussian}
\end{figure}

The privacy funnel functional is much less studied even for the simple case of jointly Gaussian. Solving the optimization in $\pf$ over $P_{T|X}$ without any assumptions is a difficult challenge. A natural assumption to make is that $P_{T|X}(\cdot|x)$ is Gaussian for each $x\in \X$. This leads to the following variant of $\pf$ $$\pf^\mathsf{G}(r)\coloneqq \inf_{\substack{\sigma\geq 0,\\I(X; T_\sigma)\geq r}} I(Y; T_\sigma),$$ 
where
$$T_\sigma \coloneqq X + \sigma N^\mathsf{G},$$
and 
$N^\mathsf{G}\sim \N(0, 1)$ is independent of $X$.
This formulation is tractable and can be computed in closed form for jointly Gaussian $(X,Y)$ as described in the following example.
\begin{example}\label{example:PF_Gaussian}
Let $X$ and $Y$ be jointly Gaussian with correlation coefficient $\rho$. First note that since mutual information is invariant to scaling, we may assume without loss of generality that both $X$ and $Y$ are zero mean and unit variance and hence we can write 
$X = \rho Y + \sqrt{1-\rho^2} M^\mathsf{G}$ where $M^\mathsf{G}\sim \N(0,1)$ is independent of $Y$. Consequently, we have 
\begin{equation}\label{IXT_Gaussian}
I(X; T_\sigma) = \frac{1}{2}\log\left(1+\frac{1}{\sigma^2}\right),
\end{equation}
and 
\begin{equation}\label{IYT_Gaussian}
    I(Y; T_\sigma) = \frac{1}{2}\log\left(1+\frac{\rho^2}{1-\rho^2+\sigma^2}\right).
\end{equation}
In order to ensure $I(X; T_\sigma)\geq r$, we must have $\sigma \leq  \left(e^{2r}-1\right)^{-\frac{1}{2}}$. Plugging this choice of $\sigma$ into \eqref{IYT_Gaussian}, we obtain 
\begin{equation}\label{Eq:PF_Gaussian}
    \pf^\mathsf{G}(r) = \frac{1}{2}\log\left(\frac{1}{1-\rho^2\left(1-e^{-2r}\right)}\right).
\end{equation}
\end{example}
This example indicates that for jointly Gaussian $(X,Y)$, we have $\pf^\mathsf{G}(r)=0$ if and only if  $r=0$ (thus perfect privacy does not occur) and the constraint $I(X; T_\sigma)= r$ is satisfied by a unique $\sigma$. These two properties in fact hold for all continuous variables $X$ and $Y$ with finite second moments as demonstrated in  Lemma~\ref{Lemma:StrictCon_IYT} in Appendix~\ref{Appendix_ProofSecIB_PF}. We use these properties to derive a second-order approximation of $\pf^\sG(r)$ when $r$ is sufficiently small. For the following theorem, we use $\var(U)$ to denote the variance of the random variable $U$ and $\var(U|V)\coloneqq \E[(U-\E[U|V])^2|V]$. We use $\sigma^2_X = \var(X)$ for short.

\begin{theorem}\label{Thm:PF_Approximation}
   For any pair of continuous random variables $(X,Y)$ with finite second moments, we have as $r\to 0$
   $$\pf^\mathsf{G}(r) = \eta(X,Y) r + \Delta(X,Y) r^2 + o(r^2),$$
   where 
   $\eta(X, Y) \coloneqq \frac{\var(\E[X|Y])}{\sigma_X^2}$ and 
   $$\Delta(X, Y)\coloneqq \frac{2}{\sigma^4_X}\left[\E[\var^2(X|Y)] - \sigma_X^2\E[\var(X|Y)]\right].$$
\end{theorem}
It is worth mentioning that the quantity $\eta(X, Y)$ was first defined by R\'enyi \cite{Renyi-dependence-measure} as an asymmetric measure of correlation between $X$ and $Y$. In fact, it can be shown that
$\eta(X, Y) = \sup_{f} \rho^2(X, f(Y)),$
where supremum is taken over all measurable functions $f$ and $\rho(\cdot, \cdot)$ denotes the correlation coefficient. 
As a simple illustration of Theorem~\ref{Thm:PF_Approximation}, consider jointly Gaussian $X$ and $Y$ with correlation coefficient $\rho$ for which $\pf^\sG$ was computed in Example~\ref{example:PF_Gaussian}. In this case, it can be easily verified that $\eta(X, Y) = \rho^2$ and  
$\Delta(X, Y)= -2\sigma_X^2\rho^2(1-\rho^2)$. Hence, for jointly Gaussian $(X, Y)$ with correlation coefficient $\rho$ and unit variance, we have 
$\pf^\sG(r) = \rho^2 r - 2\rho^2(1-\rho^2) r^2 + o(r^2)$. In Fig.~\ref{fig:PF_Gaussian}, we compare the approximation given in Theorem~\ref{Thm:PF_Approximation} for this particular case. 
\begin{figure}
    \centering
    \includegraphics[height =5cm, width=7cm]{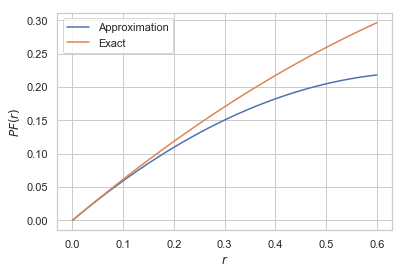}
    \caption{Second-order approximation of $\pf^\sG$ according to Theorem~\ref{Thm:PF_Approximation}  for jointly Gaussian $X$ and $Y$ with correlation coefficient $\rho=0.8$. For this particular case, the exact expression of $\pf^\sG$ is computed in \eqref{Eq:PF_Gaussian}. }
    \label{fig:PF_Gaussian}
\end{figure}

\subsection{Evaluation of $\ib$ and $\pf$}\label{Sec:Evaluation}
The constrained optimization problems in the definitions of $\ib$ and $\pf$ are usually challenging to solve numerically due to the non-linearity in the constraints. In practice, however, both $\ib$ and $\pf$ are often approximated by their corresponding Lagrangian optimizations
\begin{equation}\label{eq:Largrangian_IB}
    \L_{\ib}(\beta)\coloneqq \sup_{P_{T|X}} I(Y; T) - \beta I(X; T) = H(Y)-\beta H(X) - \inf_{P_{T|X}} \left[H(Y|T) - \beta H(X|T)\right],
\end{equation}
and 
\begin{equation}\label{eq:Largrangian_PF}
    \L_{\pf}(\beta)\coloneqq \inf_{P_{T|X}} I(Y; T) - \beta I(X; T) = H(Y)-\beta H(X) - \sup_{P_{T|X}} \left[H(Y|T) - \beta H(X|T)\right],
\end{equation}
where  $\beta\in \R_+$ is the Lagrangian multiplier that controls the tradeoff between compression and informativeness in for $\ib$ and the privacy and informativeness in $\pf$. Notice that for the computation of $\L_\ib$, we can assume, without loss of generality, that $\beta\in [0,1]$ since otherwise the maximizer of 
\eqref{eq:Largrangian_IB} is trivial. 
It is worth noting that $\L_\ib(\beta)$ and $\L_\pf(\beta)$ in fact correspond to  lines of slope $\beta$ supporting $\M$ from above and below, thereby providing a new representation of $\M$.

Let $(X', Y')$ be a pair of random variables with $X'\sim Q_X$ for some $Q_X\in \P(\X)$ and $Y'$ is the output of $P_{Y|X}$ when the input is $X'$ (i.e., $Y'\sim Q_XP_{Y|X}$). Define
$$F_\beta(Q_X)\coloneqq H(Y') - \beta H(X').$$
This function, in general, is neither convex nor concave in $Q_X$. For instance, $F(0)$ is concave and $F(1)$ is convex in $P_X$. The 
lower convex envelope (resp. upper concave envelope) of $F_\beta(Q_X)$ is defined as the largest (resp. smallest) convex (resp. concave) smaller (larger) than $F_\beta(Q_X)$. Let $\conv[F_\beta(Q_X)]$ and $\conc[F_\beta(Q_X)]$ denote the lower convex and upper concave envelopes of $F_\beta(Q_X)$, respectively.  If $F_\beta(Q_X)$ is convex at $P_X$, that is  $\conv[F_\beta(Q_X)]\big|_{P_X} = F_\beta(P_X)$, then $F_\beta(Q_X)$ remains convex at $P_X$ for all $\beta'\geq \beta$ because 
\begin{align*}
    \conv[F_{\beta'}(Q_X)] & = \conv[F_{\beta}(Q_X) - (\beta'-\beta)H(X')] \\
    & \geq  \conv[F_{\beta}(Q_X)] + \conv[- (\beta'-\beta)H(X')]\\
    & = \conv[F_{\beta}(Q_X)] - (\beta'-\beta)H(X'),
\end{align*}
where the last equality follows from the fact that $-(\beta'-\beta)H(X)$ is convex.
Hence, at $P_X$ we have 
$$ \conv[F_{\beta'}(Q_X)]\big|_{P_X} \geq \conv[F_{\beta}(Q_X)]\big|_{P_X} - (\beta'-\beta)H(X) = F_{\beta}(P_X) -  (\beta'-\beta)H(X) = F_{\beta'}(P_X).$$
Analogously, if $F_\beta(Q_X)$ is concave at $P_X$, that is  $\conc[F_\beta(Q_X)]\big|_{P_X} = F_\beta(P_X)$, then $F_\beta(Q_X)$ remains concave at $P_X$ for all $\beta'\leq \beta$.

Notice that, according to 
\eqref{eq:Largrangian_IB} and \eqref{eq:Largrangian_PF}, we can write  
\begin{equation}\label{eq:Largrangian_IB2}
    \L_\ib(\beta) = H(Y) - \beta H(X) -\conv[F_\beta(Q_X)]\big|_{P_X}, 
\end{equation}
and 
\begin{equation}\label{eq:Largrangian_PF2}
    \L_\pf(\beta) =  H(Y) - \beta H(X) -\conc[F_\beta(Q_X)]\big|_{P_X}. 
\end{equation}
In light of the above arguments, we can write $$\L_\ib(\beta) = 0,$$
for all $\beta> \beta_\ib$ where  $\beta_\ib$ is the smallest $\beta$ such that $F_\beta(P_X)$ touches $\conv[F_\beta(Q_X)]$. Similarly, 
$$\L_\pf(\beta) = 0,$$
for all $\beta< \beta_\pf$ where  $\beta_\pf$ is the largest $\beta$ such that $F_\beta(P_X)$ touches $\conc[F_\beta(Q_X)]$.
In the following theorem, we show that $\beta_\ib$ and $\beta_\pf$ are given by the values of $\ib'(0)$ and $\pf'(0)$, respectively, given in  Theorem~\ref{Thm:Derivative_IB_PF}. A similar formulae $\beta_\ib$ and $\beta_\pf$ were given in \cite{Learability_2019}.

\begin{proposition}\label{Prop:beta_ib}
We have, 
$$\beta_\ib = \sup_{Q_X\neq P_X}\frac{\kl(Q_Y\|P_Y)}{\kl(Q_X\|P_X)},$$
and 
$$\beta_\pf = \inf_{Q_X\neq P_X}\frac{\kl(Q_Y\|P_Y)}{\kl(Q_X\|P_X)}.$$
\end{proposition}
 Kim et al. \cite{Hypercontractivity_NIPS2017} have recently proposed an efficient algorithm to  estimate $\beta_\ib$ from samples of $\pxy$ involving a simple optimization problem.  This algorithm can be readily adapted for estimating $\beta_\pf$. 
Proposition~\ref{Prop:beta_ib} implies that in optimizing the Lagrangians \eqref{eq:Largrangian_IB2} and \eqref{eq:Largrangian_PF2}, we can restrict the Lagrange multiplier $\beta$, that is
\begin{equation}\label{eq:Largrangian_IB3}
    \L_\ib(\beta) = H(Y) - \beta H(X) -\conv[F_\beta(Q_X)]\big|_{P_X}, \qquad \text{for}\qquad \beta \in [0, \beta_\ib],
\end{equation}
and 
\begin{equation}\label{eq:Largrangian_PF3}
    \L_\pf(\beta) =  H(Y) - \beta H(X) -\conc[F_\beta(Q_X)]\big|_{P_X}, \qquad \text{for}\qquad \beta \in [\beta_\pf, \infty). 
\end{equation}
\begin{remark}
As demonstrated by Kolchinsky et al. \cite{kolchinsky2018caveats}, the boundary points $0$ and $\beta_\ib$ are required for the computation of $\L_\ib(\beta)$. In fact,  when $Y$ is a deterministic function of $X$, 
then \emph{only} $\beta = 0$ and $\beta = \beta_\ib$ are required to compute the $\ib$ and other values of $\beta$ are vacuous. The same argument can also be used to justify the inclusion of $\beta_\pf$ in computing $\L_\pf(\beta)$.  Note also that since $F_\beta(Q_X)$ becomes convex for $\beta>\beta_\ib$, computing $\conc[F_\beta(Q_X)]$ becomes trivial for such values of $\beta$.  
\end{remark}

\begin{remark}
Observe that the lower convex envelope of any function $f$ can be obtained by taking Legendre-Fenchel transformation (aka.\ convex conjugate) twice. Hence, one can use the existing linear-time algorithms for approximating Legendre-Fenchel transformation (e.g., \cite{Legendre_transformation_alg, Legendre_transformation_alg2}) for approximating  
$\conv[F_\beta(Q_X)]$.   
\end{remark}

Once $\L_\ib(\beta)$ and $\L_\pf(\beta)$ are computed, we can derive $\ib$ and $\pf$ via standard results in optimization (see \cite[Section IV]{Witsenhausen_Wyner} for more details): 
\begin{equation}\label{eq:Lagrangian_converse1}
    \ib(R) = \inf_{\beta\in [0, \beta_\ib]}\beta R + \L_\ib(\beta),
\end{equation}
and 
\begin{equation}\label{eq:Lagrangian_converse2}
    \pf(r) = \sup_{\beta\in [\beta_\pf, \infty]}\beta r + \L_\pf(\beta).
\end{equation}
Following the convex analysis approach outlined by Witsenhausen and Wyner \cite{Witsenhausen_Wyner}, $\ib$ and $\pf$ can be directly computed from $\L_\ib(\beta)$ and $\L_\pf(\beta)$ by observing the following.  
Suppose for some $\beta$,  $\conv[F_\beta(Q_X)]$ (resp. $\conc[F_\beta(Q_X)]$) at $P_X$ is obtained by a convex combination of points $F_\beta(Q^{i})$, $i\in [k]$ for some $Q^1, \dots, Q^k$ in $\P(\X)$, integer $k\geq 2$, and weights $\lambda_i\geq 0$ (with $\sum_{i}\lambda_i=1$).  Then $\sum_{i}\lambda_i Q^i = P_X$, and $T^*$ with properties $P_{T^*}(i) = \lambda_i$ and  $P_{X|T^* = i} = Q^i$  attains the minimum (resp. maximum) of $H(Y|T) - \beta H(X|T)$. 
Hence, $(I(X;T^*), I(Y;T^*))$ is a point on the upper (resp. lower) boundary of $\M$; implying that $\ib(R) = I(Y; T^*)$ for $R = I(X; T^*)$ (resp. $\pf(r) = I(Y; T^*)$ for $r = I(X; T^*)$). 
If for some $\beta$, $\conv[F_\beta(Q_X)]$ at $P_X$ coincides with $F_\beta[P_X]$, then this corresponds to $\L_\ib(\beta) = 0$. The same holds for $\conv[F_\beta(Q_X)]$. Thus, all the information about the functional $\ib$ (resp. $\pf$) is contained in the subset of the domain of $\conv[F_\beta(Q_X)]$ (resp. $\conc[F_\beta(Q_X)]$) over which it differs from $F_\beta(Q_X)$. 
We will revisit and generalize this approach later in Section~\ref{Sec:Family}.

We can now instantiate this for the binary symmetric case. 
Suppose $X$ and $Y$ are binary variables and $P_{Y|X}$ is binary symmetric channel with crossover probability $\delta$, denoted by $\bsc(\delta)$ and defined as 
	\begin{equation}\label{BSC}
	    \bsc(\delta)=\begin{bmatrix}
	1-\delta&\delta\\
	\delta & 1-\delta
	\end{bmatrix},
	\end{equation}
	for some $\delta\geq 0$. To describe the result in a compact fashion, we introduce the following notation: we let $h_\mathsf{b}:[0,1]\to [0,1]$ denote the binary entropy function, i.e.,  $h_\mathsf{b}(p)=-p\log p-(1-p)\log(1-p)$. Since this function is strictly increasing $[0,\frac{1}{2}]$, its inverse exists and is denoted by $h^{-1}_\mathsf{b}:[0,1]\to [0,\frac{1}{2}]$. Also, $a*b\coloneqq a(1-b)+b(1-a)$ for $a,b\in [0,1]$.

\begin{lemma}[Mr. and Mrs. Gerber's Lemma] \label{lemma: Gerber}
    For $X\sim \sBer(p)$ for $p\leq \frac{1}{2}$ and $P_{Y|X} = \bsc(\delta)$ for $\delta\geq 0$, we have 
    	\begin{equation}\label{eq:MRsGL}
    	    \mathsf{IB}(R)=h_\mathsf{b}(p*\delta) - h_\mathsf{b}\left(\delta* h_\mathsf{b}^{-1}\big(h_\mathsf{b}(p) -R\big)\right),
    	\end{equation}
    	and 
    	\begin{equation}\label{eq:MRGL}
\pf(r) = h_\mathsf{b}(p*\delta) - \alpha h_\mathsf{b}\left ( \delta * \frac{p}{z} \right ) - \bar{\alpha}h_\mathsf{b}\left ( \delta \right ),
\end{equation}
where $r = h_\mathsf{b}(p) - \alpha h_\mathsf{b}\left ( \frac{p}{z} \right )$, $z = \max\left ( \alpha, 2p \right )$, and $\alpha\in [0,1]$.
\end{lemma}	
The result in \eqref{eq:MRsGL} was proved by Wyner and Ziv \cite{Gerber} and is widely known as \textit{Mrs. Gerber's Lemma} in information theory. Due to the similarity, we refer to \eqref{eq:MRGL} as \textit{Mr. Gerber's Lemma}. 
As described above, to prove  \eqref{eq:MRsGL} and \eqref{eq:MRGL} it suffices to derive the 
convex and concave envelopes of the mapping $F_\beta:[0,1]\to \R$ given by 
\eqn{F_beta_MRGL}{F_\beta(q)\coloneqq F_\beta(Q_X) = h_\mathsf{b}(q*\delta) - \beta h_\mathsf{b}(q),}
where $q*\delta\coloneqq q\bar\delta + \delta\bar q$ is the output distribution of $\bsc(\delta)$ when the input distribution is $\sBer(q)$ for some $q\in (0,1)$. It can be verified that $\beta_\ib \leq  (1-2\delta)^2$. This function is depicted in Fig.~\ref{fig:Gerber} depending of the values of $\beta\leq (1-2\delta)^2$.
 \begin{figure}[t]
	\centering
	\begin{subfigure}[t]{0.32\textwidth}
		\includegraphics[height=4cm, width=1\columnwidth]{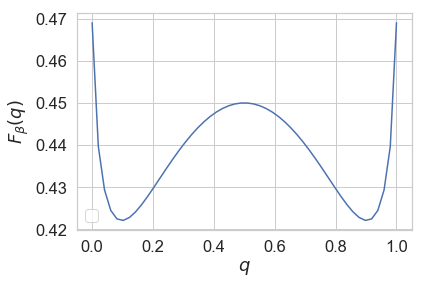}
		\caption{$\beta = 0.55$}
		\label{fig:gerber_first}
		\end{subfigure}
		\begin{subfigure}[t]{0.32\textwidth}
		\includegraphics[height=4cm, width=1\columnwidth]{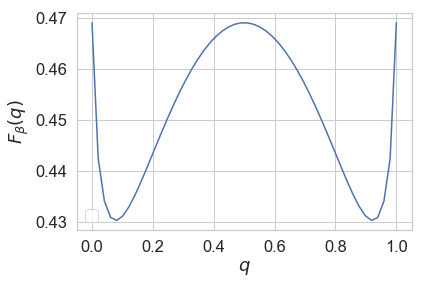}
		\caption{$\beta = 0.53$}
		\label{fig:gerber_end}
		\end{subfigure}
		\begin{subfigure}[t]{0.32\textwidth}
		\includegraphics[height=4cm, width=1\columnwidth]{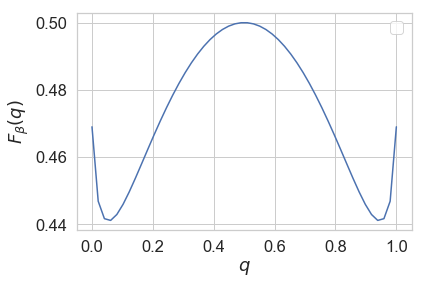}
		\caption{$\beta = 0.49$}
		\label{fig:gerber_mid}
		\end{subfigure}
		\caption{ The mapping $q\mapsto F_\beta(q)  = H(Y')-\beta H(X')$ where $X'\sim \sBer(q)$ and $Y'$ is the result of passing $X'$ through $\bsc(0.1)$, see \eqref{F_beta_MRGL}.}   
\label{fig:Gerber}
\end{figure}

	

\subsection{Operational Meaning of $\ib$ and $\pf$}\label{sec:Operational}
In this section, we illustrate several information-theoretic settings which shed light on the operational interpretation of both $\ib$ and $\pf$. 
The operational interpretation of $\ib$ has recently been extensively studied in information-theoretic settings in \cite{Bottleneck_Polyanskiy, Bottleneck_Shamai}. In particular, it was shown that $\ib$ specifies the rate-distortion region of noisy source coding problem \cite{Noisy_SourceCoding_Dobrushin, Witsenhusen_Indirect} under the logarithmic loss as the distortion measure and also the rate region of the lossless source coding with side information at the decoder \cite{Wyner_SourceCoding}. Here, we state the former setting (as it will be useful for our subsequent analysis of cardinality bound) and also provide a new information-theoretic setting in which $\ib$ appears as the solution.  Then, we describe another setting, the so-called \textit{dependence dilution}, whose achievable rate region has an extreme point specified by $\pf$.  This in fact delineate an important difference between $\ib$ and $\pf$: while $\ib$ describes the entire rate-region of an information-theoretic setup, $\pf$ specifies only a corner point of a rate region. 
Other information-theoretic settings related to $\ib$ and $\pf$ include   
CEO problem \cite{Courtade_CEO} and source coding for the Gray-Wyner network \cite{Cheuk_LI_GrayWyner}.

\subsubsection{Noisy Source Coding}
Suppose Alice has access only to a noisy version $X$ of a source of interest $Y$.  She wishes to transmit a rate-constrained description from her observation (i.e., $X$) to Bob such that he can recover $Y$ with small average distortion. 
More precisely, let $(X^n, Y^n)$ be  $n$ i.i.d. samples of $(X,Y)\sim P_{XY}$. Alice encodes her observation $X^n$ through an encoder $\phi:\X^n\to \{1, \dots, K_n\}$ and sends $\phi(X^n)$ to Bob. Upon receiving $\phi(X^n)$, Bob reconstructs a "soft" estimate of $Y^n$ via a decoder $\psi:\{1, \dots, K_n\}\to \widehat \Y^n$ where $\widehat \Y = \P(\Y)$. That is, the reproduction sequence $\hat y^n$ consists of $n$ probability measures on $\Y$. For any source and reproduction sequences $y^n$ and $\hat y^n$, respectively, the distortion is defined as 
$$d(y^n, \hat y^n) \coloneqq  \frac{1}{n}\sum_{i=1}^n d(y_i, \hat y_i),$$
where 
\begin{equation}\label{Distortion_Loss}
    d(y, \hat y) \coloneqq \log\frac{1}{\hat y(y)}.
\end{equation}
We say that a pair of rate-distortion $(\mathsf R, \mathsf D)$ is achievable if there exists a pair $(\phi, \psi)$ of encoder and decoder such that 
\begin{equation}\label{eq:Constraint_Noisy}
    \limsup_{n\to \infty} \E[d(Y^n, \psi(\phi(X^n)))]\leq \mathsf D, \qquad \text{and}\qquad \limsup_{n\to \infty}\frac{1}{n}\log K_n\leq \mathsf R.
\end{equation}
The \textit{noisy} rate-distortion function $\mathsf R^{\mathsf{noisy}}(\mathsf D)$ for a given $\mathsf D\geq 0$, is defined as the minimum rate $\mathsf R$ such that $(\mathsf R, \mathsf D)$ is an achievable rate-distortion pair.
This problem arises naturally in many data analytic problems.  Some examples include feature selection of a high-dimensional dataset, clustering, and matrix completion.  
This problem was first studied by Dobrushin and Tsybakov \cite{Noisy_SourceCoding_Dobrushin}, who showed that $\mathsf R^{\mathsf{noisy}}(\mathsf D)$ is analogous to the classical rate-distortion function   
\begin{align}
    \mathsf R^{\mathsf{noisy}}(\mathsf D) &= \inf_{\substack{P_{\hat Y|X}: \E[d(Y, \hat Y)]\leq \mathsf D,\\ Y\markov X\markov \hat Y}}I(X; \hat Y).\label{eq:rate_distortion_not_tilde} 
\end{align}
It can be easily verified that 
$\E[d(Y, \hat Y)] = H(Y|\hat Y)$ and hence (after relabeling $\hat Y$ as $T$)
\eqn{}{\mathsf R^{\mathsf{noisy}}(\mathsf D) = \inf_{\substack{P_{T|X}: I(Y; T)\geq R ,\\ Y\markov X\markov T}}I(X; T), }
where $R =  H(Y)-\mathsf D$, which is equal to $\widetilde\ib$ defined in \eqref{eq:Def_IB}.
For more details in connection between noisy source coding and $\ib$, the reader is referred to \cite{Bottleneck_Shamai, Bottleneck_Polyanskiy, Courtade_CEO, Collaborative_IB}.  
Notice that one can study an essentially identical problem where the distortion constraint \eqref{eq:Constraint_Noisy} is replaced by 
$$\lim_{n\to \infty} \frac{1}{n}I(Y^n; \psi(\phi(X^n)))\geq R, \qquad \text{and}\qquad \limsup_{n\to \infty}\frac{1}{n}\log K_n\leq \mathsf R.$$
This problem is addressed in \cite{IB_operational} for discrete alphabets $\X$ and $\Y$ and extended recently in \cite{IB_General} for any general alphabets.

\subsubsection{Test Against Independence with Communication Constraint}
As mentioned earlier, the connection between $\ib$ and noisy source coding, described above, was known and studied in \cite{Bottleneck_Shamai, Bottleneck_Polyanskiy}. Here, we provide a new information-theoretic setting which provides yet another operational meaning for $\ib$. 
Given $n$ i.i.d.\ samples $(X_1, Y_1), \dots, (X_n, Y_n)$ from joint distribution $Q$, we wish to test whether $X_i$ are independent of $Y_i$, that is, $Q$ is a product distribution. This task is formulated by the following hypothesis test:
\begin{align}\label{Hypothesis}
    \begin{aligned}
    H_0:& ~~Q=P_{XY},\\
    H_1:&~~Q=P_XP_Y,
    \end{aligned}
\end{align}
for a given joint distribution $P_{XY}$ with marginals $P_X$ and $P_Y$. Ahlswede and Csisz\'ar \cite{Hypothesis_Testing_Ahslwede} investigated this problem under a communication constraint: While $Y$ observations (i.e., $Y_1, \dots, Y_n$) are available, the $X$ observations need to be compressed at rate $R$, that is, instead of $X^n$, only $\phi(X^n)$ is present where $\phi:\X^n\to \{1, \dots, K_n\}$ satisfies 
$$\frac{1}{n}\log K_n\leq R.$$
For the type I error probability not exceeding a fixed $\eps\in (0,1)$, Ahlswede and Csisz\'ar \cite{Hypothesis_Testing_Ahslwede} derived the smallest possible type 2 error probability, defined as 
$$\beta_R(n, \eps) = \min_{\phi:\X^n\to [K]\atop \frac{1}{n}\log K_n\leq R} \min_{A\subset [K_n]\times \Y^n}\Big\{(P_{\phi(X^n)}\times P_{Y^n})(A):~~ P_{\phi(X^n)\times Y^n}(A)\geq 1-\eps\Big\}.$$
The following gives the asymptotic expression of $\beta_R(n, \eps)$ for every $\eps\in (0,1)$. For the proof, refer to \cite[Theorem 3]{Hypothesis_Testing_Ahslwede}.
\begin{theorem}[\cite{Hypothesis_Testing_Ahslwede}]
	For every $R\geq 0$ and $\eps\in (0,1)$, we have 
	\eq{\lim_{n\to \infty}-\frac{1}{n}\log \beta_R(n, \eps) = \ib(R).}
\end{theorem}

In light of this theorem, $\ib(R)$ specifies the exponential rate at which the type II error probability of the hypothesis test \eqref{Hypothesis} decays as the number of samples increases.

\subsubsection{Dependence Dilution} 
Inspired by the problems of information amplification \cite{Cover_State_Amplification} and state masking \cite{Merhav_state_masking}, Asoodeh et al. \cite{Asoode_submitted} proposed the dependence dilution setup as follows. 
Consider a source sequences $X^n$ of $n$ i.i.d. copies of $X\sim P_X$. Alice observes the source $X^n$ and wishes to encode it via
the encoder
$$f_n:\X^n\to \{1,2,\dots, 2^{nR}\},$$
for some $R>0$. The goal is to ensure that any user observing $f_n(X^n)$ can construct a list, of fixed size, of sequences in $\X^n$ that contains likely candidates of the actual sequence $X^n$ while revealing negligible information about a correlated source $Y^n$. To formulate this goal, consider the decoder   
$$g_n:\{1,2,\dots, 2^{nR}\}\to 2^{\X^n},$$ where $2^{\X^n}$ denotes the power set of $\X^n$. A \emph{dependence dilution triple} $(R, \Gamma, \Delta)\in\R^3_+$ is said to be achievable if, for any $\delta>0$, there exists a pair of encoder and decoder $(f_n, g_n)$ such that for sufficiently large $n$
\begin{equation}\label{list_decoder}
    \Pr\left(X^n\notin g_n(J)\right)<\delta,
\end{equation}
having fixed size $|g_n(J)|=2^{n(H(X)-\Gamma)},$
where $J = f_n(X^n)$ and simultaneously  
\begin{equation}\label{privacy_constraint_list_decoder}
    \frac{1}{n}I(Y^n; J)\leq \Delta+\delta.
\end{equation}
Notice that without side information $J$, the decoder can only construct a list of size $2^{nH(X)}$ which contains $X^n$ with probability close to one. However, after $J$ is observed and the list $g_n(J)$ is formed, the decoder's list size can be reduced to $2^{n(H(X)-\Gamma)}$ and thus reducing the uncertainty about $X^n$ by $ n\Gamma\in [0, nH(X)]$. This observation can be formalized to show (see \cite{Cover_State_Amplification} for details) that the constraint \eqref{list_decoder} is equivalent to  
\begin{equation}\label{amplification_equivalence}
    \frac{1}{n}I(X^n; J)\geq \Gamma-\delta,
\end{equation}
which lower bounds the amount of information $J$ carries about $X^n$. Built on this equivalent formulation, Asoodeh et al.  \cite[Corollary 15]{Asoode_submitted} derived a necessary condition for the achievable dependence dilution triple.
\begin{theorem}[\cite{Asoode_submitted}]\label{Theorem_correlation_dilution}
Any achievable dependence dilution triple $(R, \Gamma, \Delta)$ satisfies
 \begin{equation*}
  \begin{cases}
        R &\geq \Gamma\\
        \Gamma &\leq I(X;T)\\
        \Delta &\geq I(Y;T)-I(X;T)+\Gamma,
        \end{cases}
 \end{equation*}
for some  auxiliary random variable $T$ satisfying $Y\markov X\markov T$ and taking $|\T|\leq |\X|+1$ values.
\end{theorem}
According to this theorem, $\pf(\Gamma)$ specifies the best privacy performance  of the dependence dilution setup for the maximum amplification rate $\Gamma$.
While this informs the operational interpretation of $\pf$,  Theorem~\ref{Theorem_correlation_dilution} only provides an outer bound for the set of achievable dependence dilution triple $(R, \Gamma, \Delta)$. It is, however, not clear that $\pf$ characterizes the rate region of an information-theoretic setup.

The fact that $\ib$ fully characterizes the rate-region of an source coding setup has an important consequence: the cardinality of the auxiliary random variable $T$ in $\ib$ can be improved to $|\X|$ instead of $|\X|+1$.
\subsection{Cardinality Bound}\label{Sec:Cardinality_IB}
Recall that in the definition of $\ib$ in \eqref{eq:Def_IB}, no assumption was imposed on the auxiliary random variable $T$. A straightforward application of  Carath\'eodory-Fenchel-Eggleston theorem\footnote{This is a strengthening of the original Carath\'eodory theorem when the underlying space is connected, see e.g., \cite[Section III]{Witsenhausen_Convexity} or \cite[Lemma 15.4]{csiszarbook}.} reveals that $\ib$ is attained for $T$ taking values in a set $\T$ with cardinality $|\T|\leq |\X| + 1$. Here, we improve this bound and show that cardinality bound to $|\T|\leq |\X|$. 

\begin{theorem}\label{Thm:cardinality_IB}
For any joint distribution $P_{XY}$ and $R\in (0, H(X)]$, information bottleneck $\ib(R)$ is achieved by $T$ taking at most $|\X|$ values.  
\end{theorem}
The proof of this theorem hinges on the operational characterization of $\ib$ as the lower boundary of the rate-distortion region of noisy source coding problem discussed in Section~\ref{sec:Operational}. Specifically, we first show that the \textit{extreme} points of this region is achieved by $T$ taking $|\X|$ values. We then make use of a property of the  noisy source coding problem (namely, \textit{time-sharing}) to argue that \textit{all} points of this region (including the boundary points) can be attained by such $T$.   
It must be mentioned that this result was already claimed by Harremo\"es and Tishby in \cite{harremoes2007information} without  proof. 

In many practical scenarios, feature $X$ has a large alphabet. Hence, the bound $|\T|\leq |\X|$, albeit optimal, still can make the information bottleneck function computationally intractable over large alphabets. However, label $Y$ usually has a  significantly smaller alphabet. While it is in general impossible to have a cardinality bound for $T$ in terms of $|\Y|$, one can consider \textit{approximating} $\ib$ assuming $T$ takes $N$ values. The following result, recently proved by Hirche and Winter \cite{Hirche_IB_cardinality}, is in this spirit.   
\begin{theorem}[\cite{Hirche_IB_cardinality}]
For any $(X,Y)\sim P_{XY}$, we have 
$$\ib(R, N)\leq \ib(R)\leq \ib(R, N) + \delta(N),$$
where $\delta(N) = 4N^{-\frac{1}{|\Y|}}\left[\log\frac{|\Y|}{4} + \frac{1}{|\Y|}\log N\right]$ and $\ib(R, N)$ denotes the information bottleneck functional \eqref{eq:Def_IB} with the additional constraint that $|\T|\leq N$.
\end{theorem}

Recall that, unlike $\pf$,  the graph of $\ib$ characterizes the rate region of a Shannon-theoretic coding problem (as illustrated in Section~\ref{sec:Operational}), and hence any boundary points can be constructed via  time-sharing of extreme points of the rate region. This lack of  operational characterization of $\pf$ translates into a worse cardinality bound than that of $\ib$. In fact, for $\pf$ the cardinality bound $|\T|\leq |\X|+1$ cannot be improved in general. To demonstrate this,  we numerically solve the optimization in $\pf$ assuming that $|\T| = |\X|$ when both $X$ and $Y$ are binary. As illustrated in Fig.~\ref{fig:PF_NonConvexity}, this optimization does not lead to a convex function, and hence, cannot be equal to $\pf$.
\begin{figure}
    \centering
    \includegraphics[height =5cm, width=7cm]{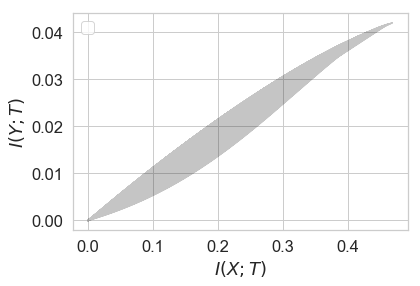}
    \caption{The set $\{(I(X; T), I(Y; T))\}$ with $P_{X} = \sBer(0.9)$, $P_{Y|X=0} = [0.9, 0.1]$, $P_{Y|X=1} = [0.85, 0.15]$, and $T$ restricted to be binary. While the upper boundary of this set is concave, the lower boundary is not convex. This implies that, unlike $\ib$, $\pf(r)$ cannot be attained by binary variables $T$. }
    \label{fig:PF_NonConvexity}
\end{figure}

\subsection{Deterministic Information Bottleneck}\label{Sec:DIB}

As mentioned earlier, $\ib$ formalizes an information-theoretic approach to clustering high-dimensional feature $X$ into cluster labels $T$ that preserve as much information about the label $Y$ as possible. The clustering label is assigned by the \textit{soft} operator $P_{T|X}$ that solves the $\ib$ formulation \eqref{eq:Def_IB} according to the rule: $X = x$ is likely assigned label $T = t$ if $\kl(P_{Y|x}\|P_{Y|t})$ is small where $P_{Y|t} = \sum_x P_{Y|x}P_{X|t}$. That is, clustering is assigned based on the similarity of conditional distributions. As in many practical scenarios, a \textit{hard} clustering operator is preferred, Strouse and Schwab \cite{strouse2017dib} suggested the following variant of $\ib$, termed as \textit{deterministic information bottleneck} $\mathsf{dIB}$
\begin{equation}\label{eq:Def_dIB}
	\mathsf{dIB}(P_{XY}, R)\coloneqq \sup_{\substack{f:\X\to \T,\\ H(f(X))\leq R}} I(Y;f(X)),
	\end{equation}
	where the maximization is taken over all deterministic functions $f$ whose range is a finite set $\T$. Similarly, one can define 
\begin{equation}\label{eq:Def_dPF}
	\mathsf{dPF}(P_{XY}, r)\coloneqq \inf_{\substack{f:\X\to \T,\\ H(f(X))\geq r}} I(Y;f(X)).
	\end{equation}	
	One way to ensure that $H(f(X))\leq R$ for a deterministic function $f$ is to restrict the cardinality of the range of $f$: if $f:\X\to [e^R]$ then $H(f(X))$ is necessarily smaller than $R$. Using this insight, we derive a lower for $\mathsf{dIB}(P_{XY}, R)$ in the following lemma.  
	\begin{lemma}\label{lemma:LB_DIB}
For any given $\pxy$, we have 
$$	\mathsf{dIB}(P_{XY}, R)\geq \frac{e^R-1}{|\X|}I(X; Y),$$
and 
$$\mathsf{dPF}(P_{XY}, r)\leq \frac{e^r-1}{|\X|} I(X; Y) + \Pr(X\geq e^r)\log\frac{1}{\Pr(X\geq e^r)}.$$
	\end{lemma}
	Note that both $R$ and $r$ are smaller than $H(X)$ and thus the multiplicative factors of $I(X; Y)$ in the lemma  are smaller than one. 
In light of this lemma, we can obtain  $$\frac{e^R-1}{|\X|}I(X; Y) \leq \ib(R)\leq I(X; Y),$$ and 	
$$\pf(r)\leq \frac{e^r-1}{|\X|}I(X; Y) + \Pr(X\geq e^r)\log\frac{1}{\Pr(X\geq e^r)}.$$
In most of practical setups, $|\X|$ might be very large, making the above lower bound for $\ib$ vacuous. In the following lemma, we partially address this issue by deriving a bound independent of $\X$ when $Y$ is binary.

\begin{lemma}\label{lemma:LB_IB_binary}
   Let $\pxy$ be a joint distribution of arbitrary $X$ and binary $Y\sim \sBer(q)$ for some $q\in (0,1)$. Then, for any $R\geq \log 5$ we have 
   $$\mathsf{dIB}(P_{XY}, R)\geq I(X; Y)-2\alpha h_\mathsf{b}\Big(\frac{I(X; Y)}{2\alpha(e^R-4)}\Big),$$
   where $\alpha = \max\{\log\frac{1}{q}, \log\frac{1}{1-q}\}$.
\end{lemma}	




\section{Family of Bottleneck Problems}
\label{Sec:Family}

In this section, we introduce a family of bottleneck problems by extending $\ib$ and $\pf$ to a large family of statistical measures. Similar to $\ib$ and $\pf$, these bottleneck problems are defined in terms of boundaries of a two-dimensional convex set induced by a joint distribution $\pxy$. Recall that $R\mapsto \ib(\pxy, R)$ and $r\mapsto \pf(\pxy, r)$ are the upper and lower boundary of the set $\M$ defined in \eqref{Def:Set_M} and expressed here again for convenience
\begin{equation}\label{Def:Set_M2}
    \M = \big\{(I(X; T), I(Y; T)):Y\markov X\markov T, (X, Y)\sim \pxy\big\}.
\end{equation}
%
Since $\pxy$ is given, $H(X)$ and $H(Y)$ are fixed. Thus, in characterizing $\M$ it is sufficient to consider only $H(X|T)$ and $H(Y|T)$. To generalize $\ib$ and $\pf$, we must therefore generalize $H(X|T)$ and $H(Y|T)$. 

Given a joint distribution $\pxy$ and two non-negative real-valued functions $\Phi:\P(\X)\to \R^+$ and $\Psi:\P(\Y)\to \R^+$, we define 
\eqn{Def:Phi}{\Phi(X|T)\coloneqq \E\left[\Phi(P_{X|T})\right] = \sum_{t\in \T}P_T(t)\Phi(P_{X|T=t}),}
and
\eqn{Def:Psi}{\Psi(Y|T)\coloneqq \E\left[\Psi(P_{Y|T})\right] = \sum_{t\in \T}P_T(t)\Psi(P_{Y|T=t}).} 
When $X\sim P_X$ and $Y\sim P_Y$, we interchangeably write $\Phi(X)$ for $\Phi(P_X)$ and $\Phi(Y)$ for $\Psi(P_Y)$. 

These definitions provide natural generalizations for Shannon's entropy and mutual information. Moreover, as we discuss later in Sections~\ref{Sec:Geussing} and \ref{Sec:Arimoto}, it also can be specialized to represent a large family of popular information-theoretic and statistical measures.
Examples include information and estimation theoretic quantities such as Arimoto's conditional entropy of order $\alpha$ for $\Phi(Q_X) = ||Q_X||_\alpha$, probability of correctly guessing for  $\Phi(Q_X) = ||Q_X||_\infty$, maximal correlation for binary case, and $f$-information for $\Phi(Q_X)$ given by $f$-divergence.    
We are able to generate a family of bottleneck problems using different instantiations of  $\Phi(X|T)$ and $\Psi(Y|T)$  in place of mutual information in $\ib$ and $\pf$. As we argue later, these problems better capture the essence of "informativeness" and "privacy"; thus providing analytical and interpretable guarantees similar in spirit to $\ib$ and $\pf$. 

Computing these bottleneck problems in general boils down to the following optimization problems
\eqn{Def:U}{\sU_{\Phi,\Psi}(\zeta)\coloneqq  \sup_{\substack{P_{T|X}:Y\markov X\markov T
\\ \Phi(X|T)\leq \zeta }} \Psi(Y|T),}
and 
\eqn{Def:L}{\sL_{\Phi,\Psi}(\zeta)\coloneqq \inf_{\substack{P_{T|X}:Y\markov X\markov T
			\\ \Phi(X|T) \geq \zeta }} \Psi(Y|T).}	
Consider the set
\begin{equation}\label{Def:Set_M2}
    \M_{\Phi, \Psi} \coloneqq \big\{(\Phi(X|T), \Psi(Y|T)):Y\markov X\markov T, (X, Y)\sim \pxy\big\}.
\end{equation}
Note that if both $\Phi$ and $\Psi$ are continuous (with respect to the total variation distance), then $\M_{\Phi, \Psi}$ is compact. Moreover, it can be easily verified that $\M_{\Phi, \Psi} $ is convex.
Hence, its upper and lower boundaries are well-defined and are characterized by the graphs of  $\sU_{\Phi,\Psi}$ and $\sL_{\Phi,\Psi}$, respectively. 
As mentioned earlier, these functional are instrumental for computing the general bottleneck problem later. Hence, before we delve into the examples of bottleneck problems, we extend the approach given in Section~\ref{Sec:Evaluation} to compute $\sU_{\Phi,\Psi}$ and $\sL_{\Phi,\Psi}$. 

\subsection{Evaluation of $\sU_{\Phi,\Psi}$ and $\sL_{\Phi,\Psi}$}\label{Sec:Evaluate_Family}
Analogous to Section~\ref{Sec:Evaluation}, we first introduce the Lagrangians of $\sU_{\Phi,\Psi}$ and $\sL_{\Phi,\Psi}$ as 
\eqn{}{\L^\sU_{\Phi, \Psi}(\beta) \coloneqq \sup_{P_{T|X}} \Psi(Y|T) - \beta \Phi(X|T),}
and 
\eqn{}{\L^\sL_{\Phi, \Psi}(\beta) \coloneqq \inf_{P_{T|X}} \Psi(Y|T) - \beta \Phi(X|T),}
where $\beta\geq 0$ is the Lagrange multiplier, respectively.  Let $(X', Y')$ be a pair of random variable with $X'\sim Q_X$ and $Y'$ is the result of passing $X'$ through the channel $P_{Y|X}$.  Letting
\eqn{}{F^{\Phi, \Psi}_\beta(Q_X) \coloneqq \Psi(Y') -\beta  \Phi(X'),}
we obtain that 
\eqn{}{\L^\sU_{\Phi, \Psi}(\beta) =  \conc[F^{\Phi, \Psi}_\beta(Q_X)]\big|_{P_X} \qquad \text{and}\qquad \L^\sL_{\Phi, \Psi}(\beta) =  \conv[F^{\Phi, \Psi}_\beta(Q_X)]\big|_{P_X},}
recalling that $\conc$ and $\conv$ are the upper concave and lower convex envelop operators.  Once we compute $\L^\sU_{\Phi, \Psi}$ and $\L^\sL_{\Phi, \Psi}$ for all $\beta\geq 0$, we can use the standard results in optimizations theory (similar to \eqref{eq:Lagrangian_converse1} and \eqref{eq:Lagrangian_converse2}) to recover $\sU_{\Phi,\Psi}$ and $\sL_{\Phi,\Psi}$. 
However, we can instead extend the approach Witsenhausen and Wyner \cite{Witsenhausen_Wyner} described in Section~\ref{Sec:Evaluation}. Suppose for some $\beta$,  $\conc[F^{\Phi, \Psi}_\beta(Q_X)]$ (resp. $\conv[F^{\Phi, \Psi}_\beta(Q_X)]$) at $P_X$ is obtained by a convex combination of points $F^{\Phi, \Psi}_\beta(Q^{i})$, $i\in [k]$ for some $Q^1, \dots, Q^k$ in $\P(\X)$, integer $k\geq 2$, and weights $\lambda_i\geq 0$ (with $\sum_{i}\lambda_i=1$).  Then $\sum_{i}\lambda_i Q^i = P_X$, and $T^*$ with properties $P_{T^*}(i) = \lambda_i$ and  $P_{X|T^* = i} = Q^i$  attains the maximum (resp. minimum) of $\Psi(Y|T) - \beta \Phi(X|T)$, implying that $(\Phi(X|T^*), \Psi(Y|T^*))$ is a point on the upper (resp. lower) boundary of $M_{\Phi, \Psi}$. Consequently, such $T^*$ satisfies $\sU_{\Phi, \Psi}(\zeta) = \Psi(Y|T^*)$ for $\zeta = \Phi(X|T^*)$ (resp. $\sL_{\Phi, \Psi}(\zeta) = \Psi(Y|T^*)$ for $\zeta = \Phi(X|T^*)$). The algorithm to compute $\sU_{\Phi, \Psi}$ and $\sL_{\Phi, \Psi}$ is then summarized in the following three steps:
\begin{itemize}
    \item Construct the functional $F^{\Phi, \Psi}_\beta(Q_X) \coloneqq \Psi(Y') -\beta  \Phi(X')$ for $X'\sim Q_X$ and $Y'\sim Q_X P_{Y|X}$ and all $Q_X\in \P(\X)$ and $\beta\geq 0$.
    \item Compute $\conc[F^{\Phi, \Psi}(Q_X)]\big|_{P_X}$ and $\conv[F^{\Phi, \Psi}(Q_X)]\big|_{P_X}$ evaluated at $P_X$. 
    \item If for distributions $Q^1, \dots , Q^k$ in $\P(X)$ for some $k\geq 1$, we have $\conc[F^{\Phi, \Psi}(Q_X)]\big|_{P_X} = \sum^k_{i=1}\lambda_iF^{\Phi, \Psi}(Q^i)$ or $\conv[F^{\Phi, \Psi}(Q_X)]\big|_{P_X} = \sum^k_{i=1}\lambda_iF^{\Phi, \Psi}(Q^i)$ for some $\lambda_i\geq 0$ satisfying $\sum_{i=1}^k\lambda_i=1$, then then $P_{X|T=i} = Q_i$, $i\in [k]$ and $P_T(i) = \lambda_i$ give the optimal $T^*$ in $\sU_{\Phi, \Psi}$ and  $\sL_{\Phi, \Psi}$, respectively.  
\end{itemize}

We will apply this approach to analytically compute $\sU_{\Phi, \Psi}$ and $\sL_{\Phi, \Psi}$ (and the corresponding bottleneck problems) for binary cases in the following sections.

\subsection{Guessing Bottleneck Problems}\label{Sec:Geussing}
Let $\pxy$ be given with marginals $P_X$ and $P_Y$ and the corresponding channel $P_{Y|X}$. Let also $Q_X\in \P(\X)$ be an arbitrary distribution on $\X$ and $Q_Y = Q_X P_{Y|X}$ be the output distribution of $P_{Y|X}$ when fed with $Q_X$.
Any channel $P_{T|X}$, together with the Markov structure $Y\markov X\markov T$, generates unique $P_{X|T}$ and $P_{Y|T}$. 
We need the following basic definition from statistics.
\begin{definition}
Let $U$ be a discrete and $V$ be an arbitrary random variables supported on $\U$ and $\V$ with $|\U|<\infty$, respectively. Then $\cp(U)$ the probability of correctly guessing $U$ and $\cp(U|V)$ the probability of correctly guessing
$U$ given $V$ are given by
\eq{\cp(U) \coloneqq \max_{u\in \U}P_U(u),}
and 
\eq{\cp(U|V) \coloneqq \max_{g}\Pr(U = g(V)) =  \E\left[\max_{u\in \U}P_{U|V}(u|V)\right].}
Moreover, the multiplicative gain of the observation $V$ in guessing $U$ is defined\footnote{The reason for $\infty$ in the notation becomes clear later.} 
$$I_\infty(U; V) \coloneqq \log\frac{\cp(U|V)}{\cp(U)}.$$
\end{definition}
As the names suggest, $\cp(U|V)$ and $\cp(U)$ characterize the optimal efficiency of guessing $U$ with or without the observation $V$, respectively. Intuitively, $I_\infty(U; V)$ quantifies how useful the observation $V$ is in estimating $U$: \textit{If it is small, then it means it is nearly as hard for an adversary  observing $V$ to guess $U$ as it is without $V$}.  This observation motivates the use of $I_\infty(Y;T)$ as a measure of privacy in lieu of $I(Y;T)$ in $\pf$.

It is worth noting that $I_\infty(U; V)$ is not symmetric in general, i.e., $I_\infty(U; V)\neq I_\infty(V; U)$. Since observing $T$ can only improve, we have $\cp(Y|T)\geq \cp(Y)$; thus $I_\infty(Y; T)\geq 0$. However, $I_\infty(Y; T) = 0$ does not necessarily imply independent of $Y$ and $T$; instead, it means $T$ is useless in estimating $Y$. As an example, consider 
$Y\sim\sBer(p)$ and $P_{T|Y=0}=\sBer(\delta)$ and $P_{T|Y=1} = \sBer(\eta)$ with $\delta, \eta\leq \frac{1}{2}< p$. Then $\cp(Y) = p$ and 
\eq{\cp(Y|T) = \max\{\bar\delta \bar p, \eta p\} + \bar \eta p.}
Thus, if $\bar\delta \bar p\leq \eta p$, then $\cp(Y|T) = \cp(Y)$. This then implies that  $I_\infty(Y;T) = 0$ whereas $Y$ and $T$ are clearly dependent; i.e, $I(Y; T)>0$. While in general $I(Y; T)$ and $I_\infty(Y; T)$ are not related, it can be shown that $I(Y; T)\leq I_\infty(Y; T)$ if $Y$ is uniform (see \cite[Proposition 1]{Asoodeh_TIT19}).  Hence, only with this uniformity assumption, $I_\infty(Y; T)$ implies the independence.

Consider $\Psi({Q_X})=-\sum_{x\in \X}{Q_X}(x)\log\left({Q_X}(x)\right)$ and $\Psi({Q_Y})=\|{Q_Y}\|_\infty$.
Clearly, we have $\Phi(X|T) = H(X|T)$. Note that 
\eqn{}{\Psi(Y|T) = \sum_{t\in \T}P_T(t)||P_{Y|T=t}||_\infty = \cp(Y|T),}
thus both measures $H(X|T)$ and $\cp(Y|T)$ are special cases of the models described in the previous section.  In particular, we can define the corresponding $\sU_{\Phi, \Psi}$ and $\sL_{\Phi, \Psi}$. We will see later that $I(X;T)$ and $\cp(Y|T)$ correspond to Arimoto's mutual information of orders $1$ and $\infty$, respectively. Define 
\eqn{}{\ib^{(\infty, 1)}(R) \coloneqq \sup_{\substack{P_{T|X}:Y\markov X\markov T \\ I(X;T)\leq R}} I_\infty(Y; T).}
This bottleneck functional formulated an interpretable guarantee: 

\vspace{0.2cm}
\centerline{\textit{$\ib^{(\infty, 1)}(R)$ characterizes the best error probability in recovering $Y$}}
\centerline{\textit{among all $R$-bit summaries of $X$}}   
\vspace{0.3cm}
\noindent Recall that the functional $\pf(r)$ aims at extracting maximum information of $X$ while protecting privacy with respect to $Y$. Measuring the privacy in terms of $\cp(Y|T)$, this objective can be better formulated by    
\eqn{}{\pf^{(\infty, 1)}(r) \coloneqq \inf_{\substack{P_{T|X}:Y\markov X\markov T \\ I(X; T)\geq r}} I_\infty(Y; T),}
with the interpretable privacy guarantee:

\vspace{0.2cm}
\centerline{\textit{$\pf^{(\infty, 1)}(r)$ characterizes the smallest probability of revealing private feature $Y$ }}
\centerline{\textit{among all representations of $X$ preserving at least $r$ bits information of $X$}}
\vspace{0.3cm}

Notice that the variable $T$ in the formulations of $\ib^{(\infty, 1)}$ and $\pf^{(\infty, 1)}$ takes values in a set $\T$ of arbitrary cardinality. However, a straightforward application of the 
 Carath\'eodory-Fenchel-Eggleston theorem (see e.g., \cite[Lemma 15.4]{csiszarbook}) reveals that the cardinality of $\T$ can be restricted to $|\X|+1$ without loss of generality. In the following lemma, we prove more basic properties of $\ib^{(\infty, 1)}$ and $\pf^{(\infty, 1)}$.
\begin{lemma}\label{Lemma:Properties_Infty}
    For any $\pxy$ with $Y$ supported on a finite set $\Y$, we have 
    \begin{itemize}
			\item $\ib^{(\infty, 1)}(0)=\pf^{(\infty, 1)}(0)=0$.
			\item $\ib^{(\infty, 1)}(R)= I_\infty(X;Y)$ for any $R\geq H(X)$ and $\pf^{(\infty, 1)}(r) = I_\infty(X; Y)$ for $r\geq H(X)$.
			\item $R\mapsto \exp(\ib^{(\infty, 1)}(R))$ is strictly increasing and concave on the range $(0, I_\infty(X;Y))$.
			\item $r\mapsto \exp(\pf^{(\infty, 1)}(r))$ is strictly increasing, and convex on the range $(0, I_\infty(X;Y))$.
		\end{itemize}
    \end{lemma}
   The proof follows the same lines as Theorem~\ref{Thm:IB_Properties} and hence omitted.  Lemma \ref{Lemma:Properties_Infty} in particular implies that inequalities $I(X; T)\leq R$ and $I(X; T)\geq r$ in the definition of $\ib^{(\infty, 1)}$ and $\pf^{(\infty, 1)}$ can be replaced by $I(X; T)= R$ and $I(X; T)= r$, respectively. 
   It can be verified that $I^\infty$ satisfies the data-processing inequality, i.e., $I^\infty(Y; T)\leq I^\infty(Y; X)$ for the Markov chain $Y\markov X\markov T$. Hence, both $\ib^{(\infty, 1)}$ and $\pf^{(\infty, 1)}$ must be smaller than $I_\infty(Y;X)$. 
   The properties listed in Lemma~\ref{Lemma:Properties_Infty} enable us to derive a slightly tighter upper bound for $\pf^{(\infty, 1)}$ as demonstrated in the following. 
   \begin{lemma}\label{Lemma:Bound_Infty}
      For any $\pxy$ with $Y$ supported on a finite set $\Y$, we have 
      \eq{\pf^{(\infty, 1)}(r)\leq \log\left[1+\frac{r}{H(X)}\Big(e^{I_\infty(Y;X)}-1\Big)\right],}
and       
      \eq{\log\left[1+\frac{R}{H(X)}\Big(e^{I_\infty(Y;X)}-1\Big)\right]\leq \ib^{(\infty, 1)}(R)\leq I_{\infty}(Y;X).}
   \end{lemma}
   The proof of this lemma (and any other results in this section) is given in Appendix~\ref{Appendix_ProofSecFamily}.
   This lemma shows that 
   the gap between $I_\infty(Y; X)$ and $\ib^{(\infty,1)}(R)$ when $R$ is sufficiently close to $H(X)$ behaves like 
   \eq{I_\infty(Y; X) - \ib^{(\infty,1)}(R)\leq I_\infty(Y; X)-\log\left[1+\frac{R}{H(X)}\Big(e^{I_\infty(Y;X)}-1\Big)\right]\approx \Big(1-e^{-I_\infty(Y;X)}\Big)\Big(1-\frac{R}{H(X)}\Big).}
Thus, $\ib^{(\infty,1)}(R)$ approaches $I_\infty(Y; X)$ as $R\to H(X)$ at least linearly.

In the following theorem, we apply the technique delineated in Section~\ref{Sec:Evaluate_Family} to derive closed form expressions for $\ib^{(\infty,1)}$ and $\pf^{(\infty,1)}$ for the binary symmetric case, thereby establishing similar results as Mr and Mrs. Gerber's Lemma.
\begin{theorem}\label{Thm:Gerber_Infty}
    For $X\sim \sBer(p)$ and $P_{Y|X} = \bsc(\delta)$ with $p, \delta\leq \frac{1}{2}$, we have 
    \eqn{Eq:PF_infty}{\pf^{(\infty, 1)}(r) = \log\left[\frac{\bar\delta - (h_\mathsf{b}(p) - r)(\frac{1}{2}-\delta)}{1-\delta*p}\right],}
    and
    \eqn{Eq:IB_infty}{\ib^{(\infty, 1)}(R) = \log\left[\frac{1-\delta*h^{-1}_\mathsf{b}(h_\mathsf{b}(p)-R)}{1-\delta*p}\right],}
    where $\bar\delta = 1-\delta$.
\end{theorem}
As described in Section~\ref{Sec:Evaluate_Family}, to compute 
$\ib^{(\infty, 1)}$ and $\pf^{(\infty, 1)}$ it suffices to derive the convex and concave envelopes of the mapping  $F_\beta^{(\infty, 1)}(q)\coloneqq \cp(Y')+\beta H(X')$ where $X'\sim \sBer(q)$ and $Y'$ is the result of passing $X'$ through $\bsc(\delta)$, i.e., $Y'\sim \sBer(\delta*q)$. In this case, $\cp(Y') = \max\{\delta*q, 1-\delta*q\}$ and $F_\beta^{(\infty, 1)}$ can be expressed as  
 \eqn{}{q\mapsto F_\beta^{(\infty, 1)}(q) = \max\{\delta*q, 1-\delta*q\}+\beta h_\mathsf{b}(q).}
This function is depicted in Fig.~\ref{fig:PF_infty}.  
 \begin{figure}[t]
	\centering
	\begin{subfigure}[t]{0.32\textwidth}
		\includegraphics[height=4cm, width=1\columnwidth]{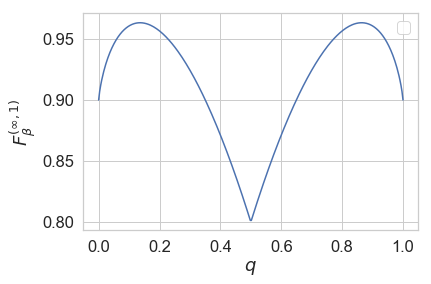}
		\caption{$\beta = 0.2$}
		\label{fig:PF_infty1}
		\end{subfigure}
		\begin{subfigure}[t]{0.32\textwidth}
		\includegraphics[height=4cm, width=1\columnwidth]{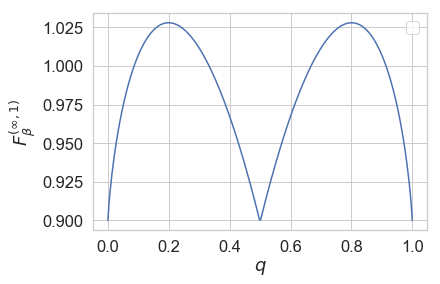}
		\caption{$\beta = 0.0.4$}
		\label{fig:PF_infty2}
		\end{subfigure}
		\begin{subfigure}[t]{0.32\textwidth}
		\includegraphics[height=4cm, width=1\columnwidth]{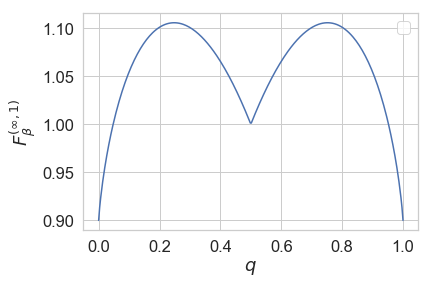}
		\caption{$\beta = 0.5$}
		\label{fig:PF_infty3}
		\end{subfigure}
		\caption{ The mapping $q\mapsto F_\beta^{(\infty, 1)}(q)  = \cp(Y')+\beta H(X')$ where $X'\sim \sBer(q)$ and $Y' \sim \sBer(q) \bsc(0.1)$.}   
\label{fig:PF_infty}
\end{figure}
The detailed derivation of convex and concave envelope of $F_\beta^{(\infty,1)}$ is given in 
Appendix~\ref{Appendix_ProofSecFamily}. The proof of this theorem also reveals the following intuitive statements. If $X\sim \sBer(p)$ and $P_{Y|X} = \bsc(\delta)$, then among all random variables $T$ satisfying $Y\markov X\markov T$ and $H(X|T) \leq \lambda$, the minimum $\cp(Y|T)$ is given by $\bar\delta - \lambda(0.5-\delta)$. 
Notice that, without any information constraint (i.e., $\lambda = 0$), $\cp(Y|T) = \cp(Y|X) = \bar\delta$. Perhaps surprisingly, this shows that the mutual information constraint has a linear effect on the privacy of $Y$.  Similarly, to prove \eqref{Eq:IB_infty}, we show that among all $R$-bit representations $T$ of $X$, the best achievable accuracy $\cp(Y|T)$ is given by $1-\delta*h^{-1}_\mathsf{b}(h_\mathsf{b}(p)-R)$.  This can be proved by combining Mrs. Gerber's Lemma (cf. Lemma~\ref{lemma: Gerber}) and Fano's inequality as follows. For all $T$ such that $H(X|T)\geq \lambda$, the minimum of $H(Y|T)$ is given by $h_\mathsf{b}(\delta*h^{-1}_\mathsf{b}(\lambda))$. Since by Fano's inequality, $H(Y|T)\leq h_\mathsf{b}(1-\cp(Y|T))$, we obtain $\delta*h^{-1}_\mathsf{b}(\lambda)\leq 1-\cp(Y|T)$
which leads to the same result as above. Nevertheless, in Appendix~\ref{Appendix_ProofSecFamily} we give another proof based on the discussion of Section~\ref{Sec:Evaluate_Family}. 

\subsection{Arimoto Bottleneck Problems}\label{Sec:Arimoto}
The bottleneck framework proposed in the last section benefited from interpretable guarantees brought forth by the quantity $I_\infty$.  In this section, we define a parametric family of statistical quantities, the so-called \textit{Arimoto's mutual information}, which includes both Shannon's mutual information and $I_\infty$ as extreme cases. 

\begin{definition}[\cite{Arimoto_Original_Paper}]\label{Def:Arimoto}
Let $U\sim P_U$ and $V\sim P_V$ be two random variables supported
over finite sets $\U$ and $\V$, respectively. Their Arimoto's mutual information of order $\alpha>1$ is defined as
\eqn{}{I_\alpha(U; V)  = H_\alpha(U) - H_\alpha(U|V),}
where 
\eqn{}{H_\alpha(U) \coloneqq \frac{\alpha}{1-\alpha}\log||P_U||_\alpha,} is the R\'enyi entropy of order $\alpha$ and \eqn{}{H_\alpha(U|V) \coloneqq \frac{\alpha}{1-\alpha}\log\sum_{v\in \V}P_V(v)||P_{U|V=v}||_\alpha,}
is the Arimoto's conditional entropy of order $\alpha$. 
\end{definition}
By continuous extension, one can define $I-\alpha(U; V)$ for $\alpha =1$ and $\alpha = \infty$ as $I(U; V)$ and $I_\infty(U; V)$, respectively. That is,
\eqn{}{\lim_{\alpha\to 1^+} I_\alpha(U; V) = I(U; V), \qquad \text{and}\qquad \lim_{\alpha\to \infty} I_\alpha(U; V) = I_\infty(U; V).}
Arimoto's mutual information was first introduced by Arimoto \cite{Arimoto_Original_Paper} and then later revisited by Liese and Vajda in \cite{Liese_fDivergence} and more recently by Verd\'u in \cite{Verdu_ALPHA}. More in-depth analysis and properties of $I_\alpha$  can be found in \cite{Arimoto's_Conditional_Entropy}. It is shown in \cite[Lemma 1]{Liao_maximal_leakage} that $I_\alpha(U; V)$ for $\alpha\in [1, \infty]$ quantifies the minimum loss in recovering $U$ given $V$ where the loss is measured in terms of the so-called $\alpha$-loss. This loss function reduces to logarithmic loss \eqref{Distortion_Loss} and $\cp(U|V)$ for $\alpha = 1$ and $\alpha = \infty$, respectively.  This sheds light on the utility and/or privacy guarantee  promised by a constraint on Arimoto's mutual information. It is now natural to use $I_\alpha$ for defining a family of bottleneck problems.

\begin{definition}
Given a pair of random variables $(X, Y)\sim \pxy$ over finite sets $\X$ and $\Y$ and $\alpha, \gamma\in[1, \infty]$, we define $\ib^{(\alpha, \gamma)}$ and $\pf^{(\alpha, \gamma)}$ as 
\eqn{}{\mathsf{IB}^{(\alpha, \gamma)}(R)\coloneqq \sup_{\substack{P_{T|X}:Y\markov X\markov T \\ I_\gamma(X;T)\leq R}} I_\alpha(Y; T),}
and 
\eqn{}{\mathsf{PF}^{(\alpha, \gamma)}(r)\coloneqq \inf_{\substack{P_{T|X}:Y\markov X\markov T \\ I_\gamma(X;T)\geq r}} I_\alpha(Y; T),}
\end{definition}
 Of course, $\mathsf{IB}^{(1, 1)}(R)  = \ib(R)$ and $\mathsf{PF}^{(1, 1)}(r)  = \pf(r)$. 
It is known that Arimoto's mutual information satisfies the data-processing inequality \cite[Corollary 1]{Arimoto's_Conditional_Entropy}, i.e., $I_\alpha(Y; T)\leq I_\alpha(Y;X)$ for the Markov chain $Y\markov X\markov T$. On the other hand, $I_\gamma(X; T)\leq H_\gamma(X)$. Thus, both $\ib^{(\alpha, \gamma)}(R)$ and $\pf^{(\alpha, \gamma)}(r)$ equal $I_\alpha(Y;X)$ for $R,r\geq H_\gamma(X)$.
Note also that  
$H_\alpha(Y|T) = \frac{\alpha}{1-\alpha}\log\Psi(Y|T)$ where $\Psi(Y|T)$ (see \eqref{Def:Psi}) corresponding to the function $\Psi(Q_Y) = ||Q_Y||_\alpha$. Consequently, $\ib^{(\alpha, \gamma)}$ and $\pf^{(\alpha, \gamma)}$ are characterized by the lower and upper boundary of $\M_{\Phi, \Psi}$, defined in \eqref{Def:Set_M2}, with respect to $\Phi(Q_X) = ||Q_X||_\gamma$ and   $\Psi(Q_Y) = ||Q_Y||_\alpha$.  
Specifically, we have  
\eqn{IB_alpha1}{\ib^{(\alpha, \gamma)}(R) = H_\alpha(Y) + \frac{\alpha}{\alpha-1}\log  \sU_{\Phi,\Psi}(\zeta), }
where $\zeta = e^{-(1-\frac{1}{\gamma})(H_\gamma(X)-R)}$, and 
\eqn{PF_alpha1}{\pf^{(\alpha, \gamma)}(r) = H_\alpha(Y) + \frac{\alpha}{\alpha-1}\log\sL_{\Phi,\Psi}(\zeta),}
where $\zeta = e^{-(1-\frac{1}{\gamma})(H_\gamma(X)-r)}$ and $\Phi(Q_X) = \|Q_X\|_\gamma$ and $\Psi(Q_Y) = \|Q_Y\|_\alpha$. This paves the way to apply the technique described in Section~\ref{Sec:Evaluation} to compute $\ib^{(\alpha, \gamma)}$ and $\pf^{(\alpha, \gamma)}$. Doing so requires the upper concave and lower convex envelope of the mapping $Q_X\mapsto \|Q_Y\|_\alpha-\beta \|Q_X\|_\gamma$ for some $\beta\geq 0$, where $Q_Y\sim Q_X P_{Y|X}$. In the following theorem, we drive these envelopes and give closed form expressions for $\ib^{(\alpha, \gamma)}$ and $\pf^{(\alpha, \gamma)}$ for a special case where $\alpha = \gamma\geq 2$.

\begin{theorem}\label{Thm:Gerber_AlphaGamma}
Let $X\sim \sBer(p)$ and $P_{Y|X} = \bsc(\delta)$ with $p, \delta\leq \frac{1}{2}$. We have for $\alpha\geq 2$  
\eq{\pf^{(\alpha, \alpha)}(r) = \frac{\alpha}{1-\alpha}\log\left[\frac{\|p*\delta\|_\alpha}{\|q*\delta\|_\alpha}\right],} 
where  $\|a\|_\alpha \coloneqq \|[a, \bar a]\|_\alpha$ for $a\in [0,1]$ and  $ q\leq p$ solves 
$$\frac{\alpha}{1-\alpha}\log\left[\frac{\|p\|_\alpha}{\|q\|_\alpha}\right] = r.$$ 
Moreover, 
\eq{\ib^{(\alpha, \alpha)}(R)=\frac{\alpha}{\alpha-1}\log\left[\frac{\bar\lambda\|\delta\|_\alpha + \lambda\|\frac{q}{z}*\delta\|_\alpha}{\|p*\alpha\|_\alpha}\right],}
where $z= \max\{2p, \lambda\}$ and $\lambda\in [0,1]$ solves
\eq{\frac{\alpha}{\alpha-1}\log\frac{\bar\lambda + \lambda\|\frac{p}{z}\|_\alpha}{\|p\|_\alpha} = R.}
\end{theorem}

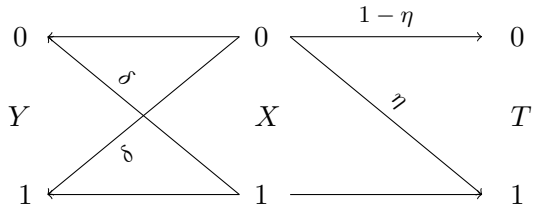
\begin{figure}[t]
		\centering
\begin{tikzpicture}[scale=1.5]
		\node (a) [circle] at (0,0) {$1~$};
		\node (b) [circle] at (0,1.4) {$0~~$};
		\node (c) [circle] at (2,0) {$~1$};
		\node (d) [circle] at (2,1.4) {$~0$};
		\node (ee1) [circle] at (2.1,0.7) {$X$};
		\node (ee2) [circle] at (-0.1,0.7) {$Y$};
		\node (ee3) [circle] at (4.35,0.7) {$T$};
		\draw[<-] (0.15,0) -- (1.85,0) node[pos=.5,sloped,below] {};
		\draw[<-] (0.15,0) -- (1.85,1.4) node[pos=.35,sloped,below] {\footnotesize{$\delta$}};
		\draw[<-] (0.15,1.4) -- (1.85,0) node[pos=.35,sloped,above] {\footnotesize{$\delta$}};
		\draw[<-] (0.15,1.4) -- (1.85,1.4) node[pos=.5,sloped,above] {};
		\node (e) [circle] at (2.3,0) {};
		\node (f) [circle] at (2.3,1.4) {};
		\node (g) [circle] at (4.27,0) {$~1$};
		\node (h) [circle] at (4.27,1.4) {$~0$};
		\draw[->] (2.3,0) -- (4,0) node[pos=.5,sloped,below] {};
		\draw[->] (2.3,1.4) -- (4,1.4) node[pos=.5,sloped,above] {\footnotesize{$1-\eta$}};
		\draw[->] (2.3,1.4) -- (4, 0) node[pos=.5,sloped,above]{\footnotesize{$\eta$}}; 
		\end{tikzpicture}	
	\caption{The structure of the optimal $P_{T|X}$ for $\pf^{(\infty, \infty)}$ when $P_{Y|X} = \bsc(\delta)$ and $X\sim \sBer(p)$ with $\delta,p\in [0, \frac{1}{2}]$. If the accuracy constraint is $\cp(X|T) \geq  \lambda$ (or equivalently $I_\infty(X|T) \geq  \log\frac{\lambda}{\bar p}$), then the  parameter of optimal $P_{T|X}$ is given by $\eta = \frac{\bar\lambda}{\bar p}$, leading to $\cp(Y|T) = \delta*\lambda$. }
\label{fig:Optimal_Filter}
\end{figure}

By letting $\alpha\to \infty$, this theorem indicates that for $X$ and $Y$ connected through $\bsc(\delta)$ and all variables $T$ forming $Y\markov X\markov T$, we have
\eqn{BSC_PF_Infty}{\cp(X|T)\geq \lambda \quad \Longrightarrow \quad \cp(Y|T)\geq \delta*\lambda,}
which can be shown to be achieved $T^*$ generated by the following channel (see Fig.~\ref{fig:Optimal_Filter})
\begin{equation}\label{Z_channel}
	    P_{T^*|X}=\begin{bmatrix}
	\frac{\lambda-p}{\bar p}&\frac{\bar\lambda}{\bar p}\\
	0 & 1
	\end{bmatrix}.
	\end{equation}
Note that, by assumption, $p\leq \frac{1}{2}$, and hence the event $\{X=1\}$ is less likely than $\{X=0\}$. Therefore, \eqref{BSC_PF_Infty} demonstrates that to ensure correct recoverability of $X$ with probability at lest $\lambda$, the most private approach (with respect to $Y$) is to obfuscate the higher-likely event $\{X=0\}$ with probability $\frac{\bar\lambda}{\bar p}$. As demonstrated in \eqref{BSC_PF_Infty} the optimal privacy guarantee is \textit{linear} in the utility parameter in the binary symmetric case. This is in fact a special case of the larger result recently proved in \cite[Theorem 1]{Asoodeh_TIT19}: the infimum of  $\cp(Y|T)$ over all variables $T$ such that $\cp(X|T)\geq \lambda$ is \textit{piece-wise linear} in $\lambda$, on equivalently, the mapping $e^{r}\mapsto \exp(\pf^{(\infty, \infty)}(r))$ is piece-wise linear. 

Computing $\pf^{(\alpha, \gamma)}$ analytically for every $\alpha, \gamma>1$ seems to be challenging, however, the following lemma provides  bounds for $\pf^{(\alpha, \gamma)}$ and $\ib^{(\alpha, \gamma)}$ in terms of $\pf^{(\infty, \infty)}$ and $\ib^{(\infty, \infty)}$, respectively. 
\begin{lemma}\label{Lemma:Alpha_Gamma_infinity}
    For any pair of random variables $(X,Y)$ over finite alphabets and $\alpha, \gamma>1$, we have 
    \eq{\frac{\alpha}{\alpha-1}\pf^{(\infty, \infty)}(f(r)) -\frac{\alpha}{\alpha-1} H_\infty(Y) + H_\alpha(Y)\leq \pf^{(\alpha, \gamma)}(r)\leq \pf^{(\infty, \infty)}(g(r)) + H_\alpha(Y)-H_\infty(Y),}
    and 
    \eq{\frac{\alpha}{\alpha-1}\ib^{(\infty, \infty)}(f(R)) -\frac{\alpha}{\alpha-1} H_\infty(Y) + H_\alpha(Y)\leq \ib^{(\alpha, \gamma)}(R)\leq \ib^{(\infty, \infty)}(g(R)) + H_\alpha(Y)-H_\infty(Y),}
    where $f(a) = \max\{a-H_\gamma(X) + H_\infty(X), 0\}$ and $g(b) = \frac{\gamma-1}{\gamma}b + H_\infty(X) - \frac{\gamma-1}{\gamma}H_\gamma(X)$.
\end{lemma}
The previous lemma can be directly applied to derive upper and lower bounds for $\pf^{(\alpha, \gamma)}$ and $\ib^{(\alpha, \gamma)}$ given  $\pf^{(\infty, \infty)}$ and $\ib^{(\infty, \infty)}$.

\subsection{$f$-Bottleneck Problems}
In this section, we describe another instantiation of the general framework introduced in terms of functions $\Phi$ and $\Psi$ that enjoys interpretable estimation-theoretic guarantee.  
\begin{definition}\label{Def:f_info}
Let $f: (0, \infty)\to \R$ be a convex function with $f(1) = 0$. Furthermore, let $U$ and $V$ be two real-valued random variables supported over $\U$ and $\V$, respectively. Their $f$-information is defined by
\eqn{Def:f_MI}{I_f(U; V) \coloneqq D_f(P_{UV}\|P_UP_V),}
where $D_f(\cdot\|\cdot)$ is the $f$-divergence \cite{Csiszar_f_divergence} between distributions and defined as 
$$D_f(P\|Q)\coloneqq \E_Q\left[f\left(\frac{\textnormal{d}P}{\textnormal{d}Q}\right)\right].$$
\end{definition}
Due to convexity of $f$, we have $D_f(P\|Q)\geq f(1) = 0$ and hence $f$-information is
always non-negative. If, furthermore, $f$ is strictly convex at $1$, then equality holds if and only $P=Q$. Csisz\'ar introduced $f$-divergence in \cite{Csiszar_f_divergence} and applied it to several problems in statistics and information theory.  More recent developments
about the properties of $f$-divergence and $f$-information can be found in \cite{Maxim_Strong_TIT} and the references therein. Any convex function $f$ with the property $f(1)=0$ results in an $f$-information. Popular examples include $f(t) = t\log t$ corresponding to Shannon's mutual information, $f(t) = |t-1|$ corresponding to $T$-information \cite{Polyanskiy}, and also $f(t) = t^2-1$ corresponding to $\chi^2$-information \cite{Hao_Privacy_estimation} for . It is worth mentioning that Arimoto's mutual information of order\footnote{Note that Arimoto's mutual information was defined in Definition~\ref{Def:Arimoto} for $\alpha>1$. However, it was similarly defined in \cite{Liese_fDivergence} for $\alpha\in (0,1)$.} $\alpha\in (0,1)$  was also shown to be an $f$-information in the binary case for a certain function $f$, see \cite[Theorem 8]{Liese_fDivergence}. 

Let $(X, Y)\sim \pxy$ be given with marginals $P_X$ and $P_Y$. Consider functions $\Phi$ and $\Psi$ on $\P(\X)$ and $\P(\Y)$ defined as 
\eq{\Phi(Q_X)\coloneqq D_f(Q_X\|P_X)\qquad \text{and}\qquad  \Psi(Q_Y)\coloneqq D_f(Q_Y\|P_Y).} 
Given a conditional distribution $P_{T|X}$, it is easy to verify that $\Phi(X|T) = I_f(X; T)$ and $\Psi(Y|T) = I_f(Y; T)$. This in turn implies that $f$-information can be utilized in \eqref{Def:U} and \eqref{Def:L} to define general bottleneck: Let $f:(0, \infty)\to \R$ and $g:(0, \infty)\to \R$ be two convex functions satisfying $f(1) = g(1) = 0$. Then we define   \eqn{}{\ib^{(f,g)}(R)\coloneqq \sup_{\substack{P_{T|X}:Y\markov X\markov T \\ I_g(X;T)\leq R}} I_f(Y; T),}
and 
\eqn{}{\pf^{(f,g)}(r)\coloneqq \inf_{\substack{P_{T|X}:Y\markov X\markov T \\ I_g(X;T)\geq r}} I_f(Y; T).}

In light of the discussion in Section~\ref{Sec:Evaluate_Family}, the optimization problems in $\ib^{(f, g)}$ and $\ib^{(f, g)}$ can be analytically solved by determining the upper concave and lower convex envelope of the mapping 
\eqn{F_beta_F}{Q_X\mapsto \F_\beta^{(f, g)}\coloneqq D_f(Q_Y\|P_Y) - \beta D_g(Q_X\|P_X),}
where $\beta \geq 0$ is the Lagrange multiplier and $Q_Y = Q_X P_{Y|X}$.

Consider the function $f_\alpha(t) = \frac{t^\alpha-1}{\alpha-1}$ with $\alpha\in (1, \infty)\cup (1, \infty)$. The corresponding $f$-divergence is sometimes called \textit{Hellinger} divergence of order $\alpha$, see e.g., \cite{Sason_f_divergence}. Note that \textit{Hellinger} divergence of order $2$ reduces to $\chi^2$-divergence. Calmon et al. \cite{Calmon_principal_TIT} and Asoodeh et al. \cite{Asoode_MMSE_submitted} showed that if $I_{f_2}(Y; T)\leq \eps$ for some $\eps\in (0,1)$, then the minimum mean-squared error (MMSE) of reconstructing any zero-mean unit-variance \textit{function} of $Y$ given $T$ is lower bounded by $1-\eps$, i.e., no function of $Y$ can be reconstructed with small MMSE given an observation
of $T$. 
This result serves a natural justification for $I_{f_2}$ as an operational measure of both privacy and utility in a bottleneck problem. 

Unfortunately, our approach described in Section~\ref{Sec:Evaluate_Family} cannot be used to compute $\ib^{(f_2, f_2)}$ or $\pf^{(f_2, f_2)}$ in the binary symmetric case. The difficulty lies in the fact that the function $\F_\beta^{f_2,f_2}$, defined in \eqref{F_beta_F}, for the binary symmetric case is either convex or concave on its entire domain depending on the value of $\beta$. Nevertheless, one can consider Hellinger divergence of order $\alpha$ with $\alpha\neq 2$ and then apply our approach to compute $\ib^{(f_\alpha, f_\alpha)}$ or $\pf^{(f_\alpha, f_\alpha)}$. Since $D_{f_2}(P\|Q)\leq (1+(\alpha-1)D_{f_\alpha}(P\|Q))^{1/(\alpha-1)}-1$ (see \cite[Corollary 5.6]{Sharp_inequalities_f}, one can justify $I_{f_\alpha}$ as a measure of privacy and  utility in a similar way as $I_{f_2}$. 

We end this section by a remark about estimating the measures studied in this section. While we consider information-theoretic regime where the underlying distribution $\pxy$ is known, in practice only samples $(x_i, y_i)$ are given. Consequently, the \textit{de facto} guarantees of bottleneck problems might be considerably different from those shown in this work. 
It is therefore essential to asses the guarantees of bottleneck problems when accessing only samples. To do so, one must derive bounds on the discrepancy between $\cp$, $I_\alpha$, and $I_f$ computed on the empirical distribution  and the true (unknown) distribution. These bounds can then be used to shed light on the de facto guarantee of the bottleneck problems. Relying on  \cite[Theorem 1]{Diaz_Robustness}, one can obtain that the gaps between the measures $\cp$, $I_\alpha$, and $I_f$ computed on empirical distributions and the true one scale as $O(1/\sqrt{n})$ where $n$ is the number of samples. This is in contrast with mutual information for which the 
similar upper bound scales as $O(\log n/\sqrt{n})$ as shown in \cite{Shamir_IB}.
Therefore, the above measures appear to be easier to estimate than mutual information.

\section{Summary and Concluding Remarks}
Following the recent surge in the use of information bottleneck ($\ib$) and privacy funnel ($\pf$) in developing and analyzing machine learning models, we investigated the functional properties of  these two optimization problems. Specifically, we showed that $\ib$ and $\pf$ correspond to the upper and lower boundary of a two-dimensional convex set $\M = \{(I(X; T), I(Y; T)): Y\markov X\markov T\}$ where $(X, Y)\sim \pxy$ represents the observable data $X$ and target feature $Y$ and the auxiliary random variable $T$ varies over all possible choices satisfying the Markov relation $Y\markov X\markov T$. This unifying  perspective on $\ib$ and $\pf$ allowed us  to adapt the classical technique of Witsenhausen and Wyner \cite{Witsenhausen_Wyner} devised for computing $\ib$ to be applicable for $\pf$ as well. We illustrated this by deriving a closed form expression for $\pf$ in the binary case --- a result reminiscent of the \textit{Mrs. Gerber's Lemma} \cite{Gerber} in information theory literature.  We then showed that both $\ib$ and $\pf$ are closely related to several information-theoretic coding problems such as noisy random coding, hypothesis testing against independence, and  dependence dilution. While these connections were partially known in previous work (see e.g., \cite{Bottleneck_Shamai, Bottleneck_Polyanskiy}), we show that they lead to an improvement on the cardinality of $T$  for computing $\ib$.  We then turned our attention to the continuous setting where $X$ and $Y$ are continuous random variables. Solving the optimization problems in $\ib$ and $\pf$ in this case without any further assumptions seems a difficult challenge in general and leads to theoretical results only when $(X,Y)$ is jointly Gaussian.  Invoking recent results on the entropy power inequality \cite{EPI_Courtade} and strong data processing inequality  \cite{calmon2015strong}, we  obtained tight bounds on $\ib$ in two different cases: (1) when $Y$ is a Gaussian perturbation of $X$ and (2) when $X$ is a Gaussian perturbation of $Y$.   We also utilized the celebrated I-MMSE relationship \cite{MMSE_Guo} to derive a second-order approximation of $\pf$ when $T$ is considered to be a Gaussian perturbation of $X$. 

In the second part of the paper, we argue that the choice of (Shannon's) mutual information in both $\ib$ and $\pf$ does not seem to carry specific operational significance. It does, however, have a desirable practical consequence: it leads to self-consistent equations \cite{tishby2000information} that can be solved iteratively (without any guarantee to convergence though). In fact, this property is unique to   mutual information among other existing information measures \cite{harremoes2007information}. Nevertheless, we argued that other information measures might lead to better interpretable guarantee for both $\ib$ and $\pf$. For instance, statistical accuracy in $\ib$ and privacy leakage in $\pf$ can be shown to be \textit{precisely} characterized by probability of correctly guessing (aka Bayes risk) or minimum mean-squared error (MMSE). Following this observation,  we introduced a large family of optimization problems, which we call \textit{bottleneck problems},  by replacing mutual information in $\ib$ and $\pf$ with Arimoto's mutual information\cite{Arimoto_Original_Paper} or $f$-information \cite{Maxim_Strong_TIT}. Invoking results from \cite{Shamir_IB, Diaz_Robustness}, we also demonstrated that these information measures are in general easier to estimate from data than mutual information. Similar to $\ib$ and $\pf$, the bottleneck problems  were shown to be fully characterized by boundaries of  a two-dimensional convex set parameterized by two real-valued non-negative functions $\Phi$ and $\Psi$.  This perspective enabled us to generalize the technique used to compute $\ib$ and $\pf$ for evaluating bottleneck problems. Applying this technique to the binary case, we derived closed form expressions for several bottleneck problems.

\appendices
\section{Proofs from Section~\ref{Sec:IB_PF}}
\label{Appendix_ProofSecIB_PF}

	 \begin{proof}[Proof of Theorem~\ref{Thm:IB_Properties}]
	  \begin{itemize}
	      \item Note that $R= 0$ in optimization problem \eqref{eq:Def_IB} implies that $X$ and $T$ are independent. Since $Y, X$ and $T$ form Markov chain $Y\markov Y\markov T$, independent of $X$ and $T$ implies independence of $Y$ and $T$ and thus $I(Y; T) = 0$. Similarly for $\pf(0)$.
	      \item Since $I(X; T)\leq H(X)$ for any random variable $T$, we have $T = X$ satisfies the information constraint $I(X; T)\leq R$ for $R\geq H(X)$. Since $I(Y; T)\leq I(Y; X)$, this choice is optimal. 
	      Similarly for $\pf$, the constraint $I(X; T)\geq r$ for $r\geq H(X)$ implies $T = X$. Hence, $\pf(r) = I(Y; X)$. 
	      \item The upper bound on $\ib$ follows from the data processing inequality: $I(Y; T)\leq \min\{I(X; T), I(X; Y)\}$ for all $T$ satisfying the Markov condition $Y\markov X\markov T$.
	      \item To prove the lower bound on $\pf$, note that  
	      \eq{I(Y; T) = I(X; T) - I(X; T|Y)\geq I(X; T) - H(X|Y).}
	  \item The concavity of $R\mapsto \ib(R)$ follows from the fact it is the upper boundary of the convex set $\M$, defined in \eqref{Def:Set_M}. This in turn implies the continuity of $\ib(\cdot)$. Monotonicity of $R\mapsto \ib(R)$ follows from the definition. Strict monotonicity follows from the convexity and the fact that $\ib(H(X)) = I(X;Y)$. 
	      \item Similar as above.
	      \item The differentiability of the map $R\mapsto \ib(R)$ follows from \cite[Lemma 6]{IB_operational}. This result in fact implies the differentiability of the map $r\mapsto \pf(r)$ as well.  
	      Continuity of the derivative of $\ib$ and $\pf$ on $(0, H(X))$ is a straightforward application of \cite[Theorem 25.5]{rock}.
	  \item Monotonicity of mappings $R\mapsto \frac{\ib(R)}{R}$ and $r\mapsto \frac{\pf(r)}{r}$ follows from the concavity and convexity of $\ib(
	  \cdot)$ and $\pf(\cdot)$, respectively. 
	  \item Strict monotonicity of $\ib(
	  \cdot)$ and $\pf(\cdot)$ imply that the optimization problems in \eqref{eq:Def_IB} and \eqref{eq:Def_PF} occur when the inequality in the constraints becomes equality.  
	  \end{itemize} 
	 \end{proof}
 
	 \begin{proof}[Proof of Theorem~\ref{Thm:Bounds_IB_PF}]
	  Recall that, according to Theorem~\ref{Thm:IB_Properties}, the mappings $R\mapsto \ib(R)$ and $r\mapsto \pf(r)$ are concave and convex, respectively. This implies that $\ib(R)$ (resp. $\pf(r)$) lies above (resp. below) the chord connecting  $(0, 0)$ and $(H(X, I(X; Y))$. This proves the lower bound (resp. upper bound) $\ib(R)\geq R\frac{I(X; Y)}{H(X)}$ (resp. $\pf(r)\leq r\frac{I(X; Y)}{H(X)}$).  
	  
	  In light of the convexity of $\pf$ and monotonicity of $r\mapsto \frac{\pf(r)}{r}$, we can write  
	  $$\frac{\pf(r)}{r}\geq\lim_{r\to 0^+}\frac{\pf(r)}{r} =  \inf_{r\neq 0}\frac{\pf(r)}{r} = \inf_{\substack{P_{T|X}:Y\markov X\markov T\\I(X; T)\neq 0}}\frac{I(Y; T)}{I(X; T)}= \inf_{\P(\X)\ni Q_X\neq P_X}\frac{\kl(Q_Y\|P_Y)}{\kl(Q_X\|P_X)},$$
	  where the last equality is due to \cite[Lemma 4]{Calmon_fundamental-Limit} and $Q_Y$ is the output distribution of the channel $P_{Y|X}$ when the input is distributed according to $Q_X$. Similarly, we can write 
	  $$\frac{\ib(R)}{R}\leq\lim_{R\to 0^+}\frac{\ib(R)}{R} =  \sup_{R\neq 0}\frac{\ib(R)}{R} = \sup_{\substack{P_{T|X}:Y\markov X\markov T\\I(X; T)\neq 0}}\frac{I(Y; T)}{I(X; T)}= \sup_{\P(\X)\ni Q_X\neq P_X}\frac{\kl(Q_Y\|P_Y)}{\kl(Q_X\|P_X)},$$
	  where the last equality is due to \cite[Theorem 4]{anantharam}.
	 \end{proof}

\begin{proof}[Proof of Theorem~\ref{Thm:IB_Additivity}]
   	Let $T_n$ be an optimal summeries of $X^n$, that is, it satisfies $T_n\markov X^n\markov Y^n$ and $I(X^n; T_n)=nR$. We can write
   	\begin{equation*}
   	I(X^n, T_n)=H(X^n)-H(X^n|T_n) = \sum_{k=1}^n\left[H(X_k)-H(X_k|X^{k-1}, T_n)\right]=\sum_{k=1}^nI(X_k; X^{k-1}, T_n),
   	\end{equation*}
   	and hence, if $R_k\coloneqq I(X_k; X^{k-1}, T_n)$, then we have 
   	\begin{equation}\label{proof_Con}
   	R=\frac{1}{n}\sum_{k=1}^nR_k.
   	\end{equation} 
   	We can similarly write
   	\begin{eqnarray*}
   	I(Y^n, T_n)&=& H(Y^n)-H(Y^n|T_n) = \sum_{k=1}^n\left[H(Y_k)-H(Y_k|Y^{k-1}, T_n)\right]\\
   	&\leq& \sum_{k=1}^n\left[H(Y_k)-H(Y_k|Y^{k-1}, X^{k-1}, T_n)\right]\\
   	&=& \sum_{k=1}^n\left[H(Y_k)-H(Y_k|X^{k-1}, T_n)\right]=\sum_{k=1}^nI(Y_k;X^{k-1}, T_n).
  	\end{eqnarray*}
  Since we have $(T_n, X^{k-1})\markov X_k\markov Y_k$ for every $k\in [n]$, we conclude from the above inequality that 
  \begin{equation}
  	I(Y^n, T_n)\leq  \sum_{k=1}^nI(Y_k;X^{k-1}, T_n) \leq  \sum_{k=1}^n \mathsf{IB}(P_{XY}, R_k)\leq n\mathsf{IB}(P_{XY}, R), 
  \end{equation}
  where the last inequality follows from concavity of the map $x\mapsto \mathsf{IB}(P_{XY}, x)$ and \eqref{proof_Con}. Consequently, we obtain
  \begin{equation}\label{proof_Direction}
  \mathsf{IB}(P_{X^nY^n}, nR)\leq n\mathsf{IB}(P_{XY}, R).
  \end{equation}
  To prove the other direction, let $P_{T|X}$ be an optimal channel in the definition of $\ib$, i.e.,  $I(X; T)=R$ and $\mathsf{IB}(P_{XY}, R)=I(Y; T)$.  Then using this channel $n$ times for each pair $(X_i, Y_i)$, we obtain $T^n=(T_1, \dots, T_n)$ satisfying $T^n\markov X^n\markov Y^n$. Since $I(X^n; T^n)=nI(X; T)=nR$ and $I(Y^n; T^n)=nI(Y; T)$, we have $ \mathsf{IB}(P_{X^nY^n}, nR)\geq n\mathsf{IB}(P_{XY}, R).$ This, together with \eqref{proof_Direction}, concludes the proof.

\end{proof}

\begin{proof}[Proof of Theorem~\ref{Thm:Derivative_IB_PF}]
       First notice that 
       $$\lim_{R\to 0}\frac{\ib(R)}{R} = \sup_{R>0}\frac{\ib(R)}{R} = \sup_{\substack{P_{T|X}:\\Y\markov X\markov T}}\frac{I(Y; T)}{I(X; T)} =  \sup_{Q_X\in \P(\X)\atop Q_X\neq P_X}\frac{\kl(Q_Y\|P_Y)}{\kl(Q_X\|P_X)},$$
       where the last equality is due to \cite[Theorem 4]{anantharam}.
       Similarly, 
       $$\lim_{r\to 0}\frac{\pf(r)}{r} =\inf_{r>0}\frac{\pf(r)}{r} = \inf_{\substack{P_{T|X}:\\Y\markov X\markov T}}\frac{I(Y; T)}{I(X; T)} = \inf_{Q_X\in \P(\X)\atop Q_X\neq P_X}\frac{\kl(Q_Y\|P_Y)}{\kl(Q_X\|P_X)},$$
       where the last equality is due to \cite[Lemma 4]{Calmon_fundamental-Limit}.
       
       Fix $x_0\in \X$ with $P_X(x_0)>0$ and let $T$ be a Bernoulli random variable specified by the following channel 
       $$P_{T|X}(1|x) = \delta 1_{\{x = x_0\}},$$
       for some $\delta>0$. This channel induces  $T \sim\sBer(\delta P_{X}(x_0))$,  
       $P_{Y|T}(y|1) = P_{Y|X}(y|x_0)$, and 
       $$P_{Y|T}(y|0) = \frac{P_Y(y) - \pxy(x_0, y)}{1-\delta P_X(x_0)}.$$
       It can be verified that
       $$I(X; T) = -\delta P_X(x_0)\log P_X(x_0) + o(\delta),$$
       and 
       $$I(Y; T) = \delta P_X(x_0)\kl(P_{Y|X}(\cdot|x_0)\|P_Y(\cdot)) + o(\delta).$$
       Setting $$\delta = \frac{r}{-P_X(x_0)\log P_X(x_0)},$$ we obtain 
       $$I(Y; T) =  \frac{\kl(P_{Y|X}(\cdot|x_0)\|P_Y(\cdot))}{-\log P_X(x_0)}r + o(r),$$
       and hence 
       $$\pf(r)\leq \frac{\kl(P_{Y|X}(\cdot|x_0)\|P_Y(\cdot))}{-\log P_X(x_0)} r + o(r).$$
Since $x_0$ is arbitrary, the result follows.  The proof for $\ib$ follows similarly.     
\end{proof}
\begin{proof}[Proof of Lemma~\ref{lem:pf_BEC}]
	 When $Y$ is an erasure of $X$, i.e., $\Y = \X\cup \{\perp\}$ with $P_{Y|X}(x|x) = 1-\delta$ and $P_{Y|X}(\perp|x) = \delta$, it is straightforward to verify that $\kl(Q_Y\|P_Y) = (1-\delta)\kl(Q_X\|P_X)$ for every $P_X$ and $Q_X$ in $\P(X)$. Consequently, we have 
	 $$\inf_{Q_X\neq P_X }\frac{\kl(Q_Y\|P_Y)}{\kl(Q_X\|P_X)} = \sup_{Q_X\neq P_X }\frac{\kl(Q_Y\|P_Y)}{\kl(Q_X\|P_X)} = 1-\delta.$$ 
	 Hence, Theorem~\ref{Thm:Bounds_IB_PF} gives the desired result. 
	 
	 To prove the second part, i.e., when $X$ is an erasure of $Y$, we need an improved upper bound of $\pf$. 	 Notice that if perfect privacy occurs for a given $\pxy$, then the upper bound for $\pf(r)$ in Theorem~\ref{Thm:Bounds_IB_PF} can be improved: \begin{equation}\label{eq:pf_UB_Improved}
	     \pf(r)\leq (r-r_0)\frac{I(X; Y)}{H(X)-r_0},
	 \end{equation}
	 where $r_0$ is the largest $r\geq 0$ such that $\pf(r) = 0$.
	 Here, we show that $r_0 = H(X|Y)$. 
	 This suffices to prove the result as  \eqref{eq:pf_UB_Improved}, together with Theorem~\ref{Thm:IB_Properties},
	 we have 
	 $$\max\{r - H(X|Y), 0\}\leq \pf\leq (r-r_0)\frac{I(X; Y)}{H(X)-r_0} = (r-H(X|Y)).$$
	 To show that $\pf(H(X|Y)) = 0$, consider the channel $P_{T|X}(t|x) = \frac{1}{|\Y|}1_{\{t\neq \perp, x\neq \perp\}}$ and  $P_{T|X}(\perp|\perp) = 1.$ It can be verified that this channel induces $T$ which is independent of $Y$ and that 
	 $$I(X; T) = H(T) - H(T|X) = H\left(\frac{1-\delta}{|\Y|}, \dots, \frac{1-\delta}{|\Y|}, \delta\right) - (1-\delta)\log |\Y| = h_{\mathsf b}(\delta) = H(X|Y),$$
	 where $h_{\mathsf b}(\delta) \coloneqq -\delta\log \delta - (1-\delta)\log(1-\delta)$ is the binary entropy function. 
	 \end{proof}

\begin{proof}[Proof of Lemma~\ref{lemma: Gerber}]
    As mentioned earlier, \eqref{eq:MRsGL} was proved in \cite{Gerber}. We thus give a proof only for \eqref{eq:MRGL}. 
    
    Consider the problem of minimizing the Lagrangian $\L_\pf(\beta)$ \eqref{eq:Largrangian_PF3} for $\beta\geq \beta_\pf$.  
    Let $X'\sim Q_X = \sBer(q)$ for some $q\in (0,1)$ and $Y'$ be the result of passing $X'$ through $\bsc(\delta)$, i.e., $Y'\sim \sBer(q*\delta)$. Recall that $F_\beta(q)\coloneqq F_\beta(Q_X) = h_\mathsf{b}(q*\delta) - \beta h_\mathsf{b}(q)$. It suffices to compute $\conc[F_\beta(q)]$ the upper concave envelope of $q\mapsto F_\beta(q)$.  It can be verified that
    $\beta_\ib \leq  (1-2\delta)^2$ and hence for all $\beta\geq (1-2\delta)^2$, $\conc[F_\beta(q)] = F_\beta(0)$.
    A straightforward computation shows that  $F_\beta(q)$ is symmetric around $q = \frac{1}{2}$ and is also concave in a region around $q = \frac{1}{2}$, where it reaches its local maximum. Hence, if $\beta$ is such that
    \begin{itemize}
        \item $F_\beta(\frac{1}{2}) < F_\beta(0)$ (see Fig.~\ref{fig:gerber_first}), then $\conc[F_\beta(q)]$ is given by the convex combination of $F_\beta(0)$ and $F_\beta(1)$.
        \item $F_\beta(\frac{1}{2}) = F_\beta(0)$ (see Fig.~\ref{fig:gerber_end}), then $\conc[F_\beta(q)]$ is given by the convex combination of $F_\beta(0)$ and $F_\beta(\frac{1}{2})$ and $F_\beta(1)$. 
        \item $F_\beta(\frac{1}{2}) > F_\beta(0)$ (see Fig.~\ref{fig:gerber_mid}), then there exists $q_\beta\in [0,\frac{1}{2}]$ such that for $q\leq q_\beta$,  $\conc[F_\beta(q)]$ is given by the convex combination of $F_\beta(0)$ and $F_\beta(q_\beta)$.
    \end{itemize} 
    Hence, assuming $p\leq \frac{1}{2}$, we can construct $T^*$ that maximizes $H(Y|T)-\beta H(X|T)$ in three different cases corresponding three cases above:
    \begin{itemize}
        \item In the first case, $T^*$ is binary and we have $P_{X|T^* = 0} = \sBer(0)$ and $P_{X|T^* = 1} = \sBer(1)$ with $P_{T^*} = \sBer(p)$.
        \item In the second case, $T^*$ is ternary and we have 
        $P_{X|T^* = 0} = \sBer(0)$,  $P_{X|T^* = 1} = \sBer(1)$, and  $P_{X|T^* = 2} = \sBer( \frac{1}{2})$ with $P_{T^*}=(1-p-\frac{\alpha}{2}, p-\frac{\alpha}{2}, \alpha)$ for some $\alpha\in [0, 2p]$.
        \item In the third case, $T^*$ is again binary and we have  $P_{X|T^*=0} = \sBer(0)$ and $P_{X|T^*=1} = \sBer(\frac{p}{\alpha})$ with $P_{T^*} = \sBer(\alpha)$ for some $\alpha\in [2p, 1]$. 
    \end{itemize}
    Combining these three cases, we obtain the result in \eqref{eq:MRGL}.
\end{proof}
\begin{proof}[Proof of Lemma~\ref{Lem:IB_UB_Courtade}]
       Let $X = Y + \sigma N^\mathsf{G}$ where $\sigma>0$ and $N^\mathsf{G}\sim \N(0,1)$ is independent of $Y$. 
 According to the improved entropy power inequality proved in \cite[Theorem 1]{EPI_Courtade}, we can write 
       $$e^{2(H(X)-I(Y; T))}\geq e^{2(H(Y)-I(X; T))} + 2\pi e \sigma^2,$$
       for any random variable $T$ forming $Y\markov X\markov T$. This, together with Theorem~\ref{Thm:IB_Additivity}, implies the result.  
\end{proof}
	 
\begin{proof}[Proof of Corollary~\ref{Cor:Gaussian_IB}]
       Since $(X, Y)$ are jointly Gaussian, we can write $X  = Y + \sigma N^\mathsf{G}$ where $\sigma = \sigma_Y\frac{\sqrt{1-\rho^2}}{\rho}$ and $\sigma^2_Y$ is the variance of $Y$. Applying Lemma~\ref{Lem:IB_UB_Courtade} and noticing that $H(X) = \frac{1}{2}\log(2\pi e (\sigma_Y^2+\sigma^2))$, we obtain 
       \begin{equation}\label{Prof_GaussianIB}
           I(Y; T)\leq \frac{1}{2}\log\frac{\sigma^2+\sigma_Y^2}{\sigma^2+\sigma_Y^2e^{-2I(X; T)}} = \frac{1}{1-\rho^2 + \rho^2e^{-2I(X; T)}},
       \end{equation}
       for all channels $P_{T|X}$ satisfying $Y\markov X\markov T$. 
       This bound is attained by Gaussian $P_{T|X}$. Specifically, assuming $T + X + \tilde\sigma M^\mathsf{G}$  where  $\tilde\sigma^2 = \sigma_Y^2\frac{e^{-2R}}{\rho^2(1-e^{-2R})}$ for $R\geq 0$ and $M^\mathsf{G}\sim \N(0, \tilde\sigma^2)$ independent of $X$, it can be easily verified that 
       $I(X; T) = R$ and $I(Y; T) = \frac{1}{1-\rho^2 + \rho^2e^{-2R}}$. This, together with  \eqref{Prof_GaussianIB}, implies
       $\ib(R) = \frac{1}{1-\rho^2 + \rho^2e^{-2R}}.$
\end{proof}
Next, we wish to prove Theorem~\ref{Thm:PF_Approximation}. However, we need the following preliminary lemma before we delve into its proof.
\begin{lemma}\label{Lemma:StrictCon_IYT}
    Let $X$ and $Y$ be continuous correlated random variables with $\E[X^2]<\infty$ and $\E[Y^2]<\infty$. Then the mappings $\sigma\mapsto I(X; T_\sigma)$ and $\sigma\mapsto I(Y; T_\sigma)$ are continuous, strictly decreasing, and
    $$I(X; T_\sigma)\to 0, \quad \text{and}\quad I(Y; T_\sigma)\to 0 \quad \text{as}\quad \sigma\to \infty.$$
\end{lemma}
\begin{proof}
       The finiteness of $\E[X^2]$ and $\E[Y^2]$ imply that $H(X)$ and $H(Y)$ are finite. A straightforward application of the entropy power inequality (cf. \cite[Theorem 17.7.3]{Cover_Book}) implies that $H(T_\sigma)$ is also finite. Thus, $I(X; T_\sigma)$ and $I(Y; T_\sigma)$ are well-defined. According to the data processing inequality, we have $I(X; T_{\sigma+\delta})< I(X; T_{\sigma})$ for all $\delta>0$ and also $I(Y; T_{\sigma+\delta})\leq I(Y; T_{\sigma})$
       where the equality occurs if and only if $X$ and $Y$ are independent. Since, bu assumption $X$ and $Y$ correlated, it follows $I(Y; T_{\sigma+\delta})< I(Y; T_{\sigma})$. Thus, both $I(X; T_\sigma)$ and    $I(Y; T_\sigma)$ are strictly decreasing. 
       
       For the proof of continuity, we consider two cases $\sigma = 0$ and $\sigma>0$ separately. We first give the poof for $I(X; T_\sigma)$.  Since $H(\sigma N^\mathsf{G}) = \frac{1}{2}\log(2\pi e \sigma^2)$, we have $\lim_{\sigma\to 0}H(\sigma N^\mathsf{G}) = \infty$ and thus $\lim_{\sigma\to 0}I(X; T_\sigma) = \infty$ that is equal to $I(X; T_0)$. For $\sigma>0$, let ${\sigma_n}$ be a sequence of positive numbers converging to $\sigma$. In light of de Bruijn's identity (cf. \cite[Theorem 17.7.2]{Cover_Book}), we have $H(T_{\sigma_n})\to H(T_\sigma)$, implying the continuity of $\sigma\mapsto I(X; T_\sigma)$. 
       
       Next, we prove the  continuity of $\sigma\mapsto I(Y; T_\sigma)$. For the sequence of positive numbers ${\sigma_n}$ converging to $\sigma>0$, we have $I(Y; T_{\sigma_n}) = H(T_{\sigma_n})- H(T_{\sigma_n}|Y)$. We only need to show $H(T_{\sigma_n}|Y)\to H(T_{\sigma}|Y)$. Invoking again de Brujin's identity, we obtain  $H(T_{\sigma_n}|Y=y)\to H(T_{\sigma}|Y=y)$ for each $y\in \Y$. The desired result follows from dominated convergence theorem. 
       Finally, the The  continuity of $\sigma\mapsto I(Y; T_\sigma)$ when $\sigma = 0 $ follows from \cite[Page 2028]{Linder} stating that $H(T_{\sigma_n}|Y=y)\to H(X|Y=y)$ and then applying dominated convergence theorem. 
       
       Note that 
       $$0\leq I(Y; T_\sigma)\leq I(X; T_\sigma)\leq \frac{1}{2}\log\left(1+\frac{\sigma^2_{X}}{\sigma^2}\right),$$
       where $\sigma^2_X$ is the variance of $X$ and the last inequality follows from the fact that $I(X; X + \sigma N^\mathsf{G})$ is maximized when $X$ is Gaussian. Since by assumption $\sigma_X<\infty$,  it follows that both $I(X; T_\sigma)$ and $I(Y; T_\sigma)$ converge to zero as $\sigma\to \infty$.
\end{proof}
In light of this lemma, there exists a unique $\sigma\geq 0$ such that $I(X; T_\sigma) = r$. Let $\sigma_r$ denote such $\sigma$. Therefore, we have 
$\pf(r) = I(Y; T_{\sigma_r}).$
This enables us to prove Theorem~\ref{Thm:PF_Approximation}. 
\begin{proof}[Prof of Theorem~\ref{Thm:PF_Approximation}]
The proof relies on the I-MMSE relation in information theory literature. We briefly describe it here for convenience. 
Given any pair of random variables $U$ and $V$, the minimum mean-squared error (MMSE) of estimating $U$ given $V$ is given by
 \begin{equation*}
  \mmse(U|V):=\inf_{f}\E[(U-f(V))^2] = \E[\left(U-\E[U|V]\right)^2]=\E[\var(U|V)],
\end{equation*}
where the infimum is taken over all measurable functions $f$ and  $\var(U|V)=\E[(U-\E[U|V])^2|V]$. Guo et al.\ \cite{MMSE_Guo} proved the following identity, which is referred to as \textit{I-MMSE formula}, relating the input-output mutual information of the additive Gaussian channel $T_\sigma=X+\sigma N^\sG$, where $N^\sG\sim\N(0,1)$ is independent of $X$, with the MMSE of the input given the output:
  \begin{equation}\label{I_MMSE}
    \frac{\text{d}}{\text{d}(\sigma^2)}I(X;T_\sigma)=-\frac{1}{2\sigma^4}\mmse(X|T_\sigma).
  \end{equation}
 Since $Y$, $X$, and $T_\sigma$ form the Markov chain $Y\markov X\markov T_\sigma$, it follows that $I(Y; T_\sigma)=I(X; T_\sigma)-I(X; T_\sigma|Y)$. Thus, two applications of \eqref{I_MMSE} yields
     \begin{equation}\label{I_MMSE2}
    \frac{\text{d}}{\text{d}(\sigma^2)}I(Y; T_\sigma)=-\frac{1}{2\sigma^4}\left[\mmse(X|T_\sigma)-\mmse(X|T_\sigma, Y)\right].
  \end{equation}
The second derivative of  $I(X; T_\sigma)$ and $I(Y; T_\sigma)$ are also known via the formula \cite[Proposition 9]{MMSE_Guo_Wu}
  \begin{equation}\label{Second_Derivative_MMSE}
    \frac{\text{d}}{\text{d}(\sigma^2)}\mmse(X|T_\sigma)=\frac{1}{\sigma^4}\E[\var^2(X|T_\sigma)]\qquad \text{and}\qquad \frac{\text{d}}{\text{d}(\sigma^2)}\mmse(X|T_\sigma, Y)=\frac{1}{\sigma^4}\E[\var^2(X|T_\sigma, Y)].
  \end{equation}
With these results in mind, we now begin the proof. Recall that $\sigma_r$ is the unique $\sigma$ such that $I(X; T_\sigma) = r$, thus implying $\pf^\mathsf{G}(r) = I(Y; T_{\sigma_r})$.  
We have 
\begin{equation}\label{Proof_derovative_PF}
    \frac{\text{d}}{\text{d}r} \pf^\mathsf{G}(r)= \left[\frac{\text{d}}{\text{d}(\sigma^2)}I(Y; T_\sigma)\right ]_{\sigma = \sigma_r}\frac{\text{d}}{\text{d}r}\sigma^2_r.
\end{equation}
To compute the derivative of $\pf(r)$, we therefore need to compute the derivative of $\sigma_r^2$ with respect to $r$. To do so, notice that from the identity $I(X; T_{\sigma_r}) = r$ we can obtain
$$1 = \frac{\text{d}}{\text{d}r}I(X; T_{\sigma_r}) = \left[\frac{\text{d}}{\text{d}(\sigma^2)}I(X; T_{\sigma_r})\right]_{\sigma = \sigma_r}\frac{\text{d}}{\text{d}r}\sigma^2_r = -\frac{1}{2\sigma^4}\mmse(X|T_{\sigma_r})\frac{\text{d}}{\text{d}r}\sigma^2_r,$$
implying 
$$\frac{\text{d}}{\text{d}r}\sigma^2_r = \frac{-2\sigma^4}{\mmse(X|T_{\sigma_r})}.$$
Plugging this identity into \eqref{Proof_derovative_PF} and invoking \eqref{I_MMSE2}, we obtain 
\begin{equation}\label{eq:firstDerivative_PF}
   \frac{\text{d}}{\text{d}r} \pf^\mathsf{G}(r) = \frac{\mmse(X|T_{\sigma_r})-\mmse(X|T_{\sigma_r}, Y)}{\mmse(X|T_{\sigma_r})}. 
\end{equation}
The second derivative can be obtained via \eqref{Second_Derivative_MMSE}
$$\frac{\text{d}^2}{\text{d}r^2} \pf^\mathsf{G}(r) = 2\frac{\E[\var^2(X|T_{\sigma_r},Y)]}{\mmse^2(X|T_{\sigma_r})} - 2\E[\var^2(X|T_{\sigma_r})]\frac{\mmse(X|T_{\sigma_r}, Y)}{\mmse^3(X|T_{\sigma_r})}.$$
Since $\sigma_r\to \infty$ as $r\to 0$, we can write 
$$\frac{\text{d}}{\text{d}r} \pf^\mathsf{G}(r)\Big|_{r=0} = \frac{\sigma^2_X-\E[\var(X|Y)]}{\sigma_X} = \frac{\var(\E[X|Y])}{\sigma^2_X} = \eta(X, Y),$$
where $\var(\E[X|Y])$ is the variance of the conditional expectation $X$ given $Y$ and the last equality comes from the law of total variance. 
and 
$$\frac{\text{d}^2}{\text{d}r^2} \pf^\mathsf{G}(r)\Big|_{r=0} = \frac{2}{\sigma^4_X}\left[\E[\var^2(X|Y)] - \sigma_X^2\E[\var(X|Y)]\right].$$
Taylor expansion of $\pf(r)$ around $r=0$ gives the result. 
\end{proof}

\begin{proof}[Proof of Theorem~\ref{Thm:cardinality_IB}]
 The main ingredient of this proof is a result by Jana \cite[Lemma 2.2]{CardinalityBound} which provides a tight cardinality bound for the  auxiliary  random variables in the canonical problems in network information theory (including noisy source coding problem described in \ref{sec:Operational}). 
Consider  a pair of random variables $(X,Y)\sim P_{XY}$ and let $d: \Y\times \hat \Y\to \R$ be  an arbitrary  distortion measure defined for arbitrary reconstruction alphabet $\hat \Y$. 
\begin{theorem}[\cite{CardinalityBound}]\label{thm:Jana}
	Let $\A$ be the set of all pairs $(\mathsf R, \mathsf D)$ satisfying 
	\eq{I(X; T)\leq \mathsf R\qquad \qquad \text{and}\qquad \qquad \E[d(Y, \psi(T))]\leq \mathsf D,}
	for some mapping $\psi:\T\to \hat\Y$ and some joint distributions $P_{XYT} = P_{XY}P_{T|X}$. Then every extreme points of $\A$ corresponds to some choice of auxiliary variable $T$ with alphabet size $|\T|\leq |\X|$.
\end{theorem}
Measuring the distortion in the above theorem in terms of the logarithmic loss as in \eqref{Distortion_Loss}, we obtain that
$$\A=\{(\mathsf R, \mathsf D)\in \R^2_+:~ \mathsf R\geq \mathsf R^{\mathsf{noisy}}(\mathsf D)\},$$
where $R^{\mathsf{noisy}}(\mathsf D)$ is
given in \eqref{eq:rate_distortion_not_tilde}. We observed in Section~\ref{sec:Operational} that $\ib$ is fully characterized by the mapping $\mathsf D\mapsto R^{\mathsf{noisy}}(\mathsf D)$ and thus by $\A$. 
In light of Theorem~\ref{thm:Jana},  all extreme points of $\A$ are achieved by a choice of $T$ with cardinality size $|\T|\leq |\X|$. 
Let $\{(\mathsf R_i, \mathsf D_i)\}$ be the set of extreme points of $\A$ each constructed by channel $P_{T_i|X}$ and mapping $\psi_i$. Due to the convexity of $\A$, each point $(\mathsf R, \mathsf D)\in \A$ is expressed as a convex combination of $\{(\mathsf R_i, \mathsf D_i)\}$ with coefficient $\{\lambda_i\}$; that is there exists a channel $P_{T|X} = \sum_{i}\lambda_iP_{T_i|X}$ and a mapping $\psi(T) = \sum_{i}\lambda_i\psi_i(T_i)$ such that $I(X; T) = \mathsf R$ and $\E[d(Y, \psi(T))] = \mathsf D$. This construction, often termed \textit{timesharing} in information theory literature, implies that \textit{all} points in $\A$ (including the boundary points) can be achieved with a variable $T$ with $|\T|\leq |\X|$. Since the boundary of $\A$ is specified by the mapping $R\mapsto \ib(R)$, we conclude that $\ib(R)$ is achieved by a variable $T$ with cardinality $|\T|\leq |\X|$ for very $R<H(X)$. 
\end{proof}

	\begin{proof}[Proof of Lemma~\ref{lemma:LB_DIB}]
	The following proof is inspired by \cite[Proposition 1]{Polyanskiy_Distilling}.
	 Let $\X = \{1, \dots, m\}$. 
	 We sort the elements in $\X$ such that $$P_X(1) \kl(P_{Y|X=1}\|P_Y)\geq \dots\geq P_X(m) \kl(P_{Y|X=m}\|P_Y).$$ 
	 Now consider the function $f:\X\to [M]$ given by $f(x) = x$ if $x< M$ and $f(x) = M$ if $x\geq M$ where $M= e^R$. 
	 Let $Z = f(X)$. We have $P_Z(i) = P_X(i)$ if $i<M$ and $P_Z(M) = \sum_{j\geq M}P_X(j)$.
	 We can now write
	 \begin{align*}
	     I(Y; Z) &= \sum_{i=1}^{M-1} P_X(i)D(P_{Y|X=i}\|P_Y) + P_Z(M)D(P_{Y|Z=M}\|P_Y)\\
	     &\geq \sum_{i=1}^{M-1} P_X(i)D(P_{Y|X=i}\|P_Y)\\
	     & \geq \frac{M-1}{|\X|}\sum_{i\in \X} P_X(i)D(P_{Y|X=i}\|P_Y)\\
	     & = \frac{M-1}{|\X|}I(X; Y).
	 \end{align*}
Since $f(X)$ takes values in $[M]$, it follows that $H(f(X))\leq R$. Consequently, we have 
$$\mathsf{dIB}(P_{XY}, R)\geq \sup_{f:\X\to [M]} I(Y;f(X)) \geq \frac{M-1}{|\X|}I(X; Y).$$

For the privacy funnel, the proof proceeds as follows. We sort the elements in $\X$ such that $$P_X(1) \kl(P_{Y|X=1}\|P_Y)\leq \dots\leq P_X(m) \kl(P_{Y|X=m}\|P_Y).$$ 
	 Consider now the function $f:\X\to [M]$ given by $f(x) = x$ if $x< M$ and $f(x) = M$ if $x\geq M$. As before, let $Z = f(X)$. Then, we can write, 
	 \begin{align*}
	     I(Y; Z) &= \sum_{i=1}^{M-1} P_X(i)D(P_{Y|X=i}\|P_Y) + P_Z(M)D(P_{Y|Z=M}\|P_Y)\\
	     &\leq \frac{M-1}{|\X|} \sum_{i\in \X} P_X(i)D(P_{Y|X=i}\|P_Y) + P_Z(M)D(P_{Y|Z=M}\|P_Y)\\
	     &= \frac{M-1}{|\X|} I(X; Y) + P_Z(M)\sum_{y\in \Y}P_{Y|Z}(y|M)\log\frac{P_{Y|Z}(y|M)}{P_Y(y)}\\
	     &\leq \frac{M-1}{|\X|} I(X; Y) + \Pr(X\geq M)\sum_{y\in \Y}\left[\sum_{i}P_{Y|X}(y|i)\frac{P_X(i)1_{\{i\geq M\}}}{\Pr(X\geq M)}\right]\log\frac{\sum_{i}P_{Y|X}(y|i)\frac{P_X(i)1_{\{i\geq M\}}}{\Pr(X\geq M)}}{\sum_{i}P_{Y|X}(y|i)P_X(i)}\\
	     &\leq   \frac{M-1}{|\X|} I(X; Y) + \sum_{y\in \Y}\sum_{i\geq M}P_{Y|X}(y|i)P_X(i)\log\frac{1}{\Pr(X\geq M)}\\
	     & = \frac{M-1}{|\X|} I(X; Y) + \Pr(X\geq M)\log\frac{1}{\Pr(X\geq M)}
	 \end{align*}
where the last inequality is due to the log-sum inequality. 
	\end{proof}

\begin{proof}[Proof of Lemma~\ref{lemma:LB_IB_binary}]
 Employing the same argument as in the proof of \cite[Theorem 3]{Polyanskiy_Distilling}, we obtain that there exists a function $f:\X\to [M]$ such that 
 \begin{equation}\label{Eq:MI_LB}
     I(Y; f(X))\geq \eta I(X; Y)
 \end{equation} for any $\eta\in (0,1)$ and 
 $$M\leq 4 + \frac{4}{(1-\eta)\log 2}\log\frac{2\alpha}{(1-\eta)I(X; Y)}.$$
 Since $h^{-1}_\mathsf{b}(x)\leq \frac{x \log 2}{\log\frac{1}{x}}$ for all $x\in (0,1]$, it follows from above that (noticing that $I(X; Y)\leq \alpha$)
 $$M\leq 4 + \frac{I(X; Y)}{2\alpha} \frac{1}{h^{-1}_\mathsf{b}(\zeta)},$$
 where $\zeta \coloneqq \frac{(1-\eta)I(X; Y)}{2\alpha}$. Rearranging this, we obtain 
 $$h^{-1}_\mathsf{b}(\zeta)\leq \frac{I(X; Y)}{2\alpha(M-4)}.$$
 Assuming $M\geq 5$, we have $\frac{I(X; Y)}{2\alpha(M-4)}\leq \frac{1}{2}$ and hence 
 $$\zeta\leq h_\mathsf{b}\Big(\frac{I(X; Y)}{2\alpha(M-4)}\Big),$$
 implying 
 $$\eta\geq 1-\frac{2\alpha}{I(X; Y)}h_\mathsf{b}\Big(\frac{I(X; Y)}{2\alpha(M-4)}\Big).$$
Plugging this into \eqref{Eq:MI_LB}, we obtain 
$$I(Y; f(X))\geq I(X; Y)-2\alpha h_\mathsf{b}\Big(\frac{I(X; Y)}{2\alpha(M-4)}\Big).$$
As before, if $M = e^R$, then $H(f(X))\leq R$. Hence, 
$$\mathsf{dIB}(P_{XY}, R)\geq I(X; Y)-2\alpha h_\mathsf{b}\Big(\frac{I(X; Y)}{2\alpha(e^R-4)}\Big),$$
for all $R\geq \log 5$.
\end{proof}

\section{Proofs from Section~\ref{Sec:Family}}
\label{Appendix_ProofSecFamily}
\unskip
   \begin{proof}[Proof of Lemma~\ref{Lemma:Bound_Infty}]
    To prove the upper bound on $\pf^{(\infty, 1)}$, recall that  $r\mapsto e^{\pf^{(\infty, 1)}(r)}$ is convex. Thus, it lies below  the chord connecting points $(0, 0)$ and $(H(X), e^{I_\infty(X;Y)})$. 
    The lower bound on $\ib^{(\infty, 1)}$ is similarly obtained using the concavity of  $R\mapsto e^{\ib^{(\infty, 1)}(R)}$. This is achievable by an erasure channel. To see this consider the random variable $T_\delta$ taking values in $\X\cup \{\perp\}$ that is obtained by conditional distributions 
    $P_{T_\delta|X}(t|x) = \bar{\delta}I_{t=x}$ and $P_{T_\delta|X}(\perp|x) = \delta$ for some $\delta\geq 0$. It can be verified that $I(X; T_\delta) = \bar\delta H(X)$ and $\cp(Y|T_\delta) = \bar\delta \cp(Y|X) + \delta\cp(Y)$. By taking $\delta = 1-\frac{R}{H(X)}$, this channel meets the constraint $I(X; T_\delta)= R$.
    Hence, 
    \eq{\ib^{(\infty, 1)}(R)\geq \log\left[\frac{\cp(Y|T_\delta)}{\cp(Y)}\right] = \log\left[1-\frac{R}{H(X)} + \frac{R}{H(X)}\frac{\cp(Y|X)}{\cp(Y)}\right] .}
   \end{proof}
\begin{proof}[Proof of Theorem~\ref{Thm:Gerber_Infty}]
 We begin by $\pf^{(\infty,1)}$.  As described in Section~\ref{Sec:Evaluate_Family}, and similar to Mrs. Gerber's Lemma (Lemma~\ref{lemma: Gerber}), we need to construct the lower convex envelope $\conv[F_\beta^{(\infty, 1)}]$ of  $F_\beta^{(\infty, 1)}(q)=\cp(Y')+\beta H(X')$ where $X'\sim \sBer(q)$ and $Y'$ is the result of passing $X'$ through $\bsc(\delta)$, i.e., $Y'\sim \sBer(\delta*q)$. In this case, $\cp(Y') = \max\{\delta*q, 1-\delta*q\}$. Hence, we need to determine the lower convex envelope of the map 
 \eqn{}{q\mapsto F_\beta^{(\infty, 1)}(q) = \max\{\delta*q, 1-\delta*q\}+\beta h_\mathsf{b}(q).}
 A straightforward computation shows that  $F_\beta^{(\infty, 1)}(q)$ is symmetric around $q = \frac{1}{2}$ and  is also concave in $q$ on $q\in [0,\frac{1}{2}]$ for any $\beta$. Hence, $\conv[F_\beta^{(\infty, 1)}]$ is obtained as follows depending on the values of $\beta$: 
    \begin{itemize}
        \item $F_\beta^{(\infty, 1)}(\frac{1}{2}) < F_\beta^{(\infty, 1)}(0)$ (see Fig.~\ref{fig:PF_infty1}), then $\conv[F_\beta^{(\infty, 1)}]$ is given by the convex combination of $F_\beta^{(\infty, 1)}(0)$, $F_\beta^{(\infty, 1)}(1)$, and $F_\beta^{(\infty, 1)}(\frac{1}{2})$.
        \item $F_\beta^{(\infty, 1)}(\frac{1}{2}) = F_\beta^{(\infty, 1)}(0)$ (see Fig.~\ref{fig:PF_infty2}), then $\conv[F_\beta^{(\infty, 1)}]$ is given by the convex combination of $F_\beta^{(\infty, 1)}(0)$, $F_\beta^{(\infty, 1)}(\frac{1}{2})$, and $F_\beta^{(\infty, 1)}(1)$. 
        \item $F_\beta^{(\infty, 1)}(\frac{1}{2}) > F_\beta^{(\infty, 1)}(0)$ (see Fig.~\ref{fig:PF_infty3}), then $\conv[F_\beta^{(\infty, 1)}]$ is given by the convex combination of $F_\beta^{(\infty, 1)}(0)$ and $F_\beta^{(\infty, 1)}(1)$. 
    \end{itemize} 
    Hence, assuming $p\leq \frac{1}{2}$, we can construct $T^*$ that minimizes $\cp(Y|T)-\beta H(X|T)$. Considering the  first two cases, we obtain that $T^*$ is ternary with $P_{X|T^* = 0} = \sBer(0)$, $P_{X|T^* = 1} = \sBer(1)$, and $P_{X|T^* = 2} = \sBer(\frac{1}{2})$ with marginal $P_{T^*} = [1-p-\frac{\alpha}{2}, p-\frac{\alpha}{2}, \alpha]$ for some $\alpha\in [0, 2p]$. This leads to $\cp(Y|T^*) = \bar\alpha\bar\delta +\frac{1}{2}\alpha$ and $I(X;T^*) = h_\mathsf{b}(p)-\alpha$.
    Note that $\cp(Y|T^*)$ covers all possible domain $[\cp(Y), \bar\delta]$ by varying $\alpha$ on $[0, 2p]$.  Replacing $I(X;T^*)$ by $r$, we obtain $\alpha = h_\mathsf{b}(p) - r$ leading to  
    $\cp(Y|T^*) = \bar\delta - (h_\mathsf{b}(p) - r)(\frac{1}{2}-\delta)$. Since $\cp(Y) = 1-\delta*p$, the desired result follows. 
       
    To derive the expression for $\ib^{(\infty,1)}$, recall that we need to derive $\conc[F_\beta^{(\infty,1)}]$ the upper concave envelope of $F_\beta^{(\infty,1)}$. It is clear from Fig.~\ref{fig:PF_infty} that  $\conc[F_\beta^{(\infty,1)}]$ is  obtained by replacing  $F_\beta^{(\infty,1)}(q)$ on the interval $[q_\beta, 1-q_\beta]$ by its maximum value over $q$ where 
    \eq{q_\beta\coloneqq \frac{1}{1+e^{\frac{1-2\delta}{\beta}}},}
    is the maximizer of $F_\beta^{(\infty,1)}(q)$ on $[0, \frac{1}{2}]$. In other words, 
    \eq{\conc[F_\beta^{(\infty,1)}(q)] = \begin{cases} F_\beta^{(\infty,1)}(q_\beta), & \text{for}~ q\in [q_\beta, 1-q_\beta],\\  
    F_\beta^{(\infty,1)}(q), & \text{otherwise.}
    \end{cases}
    }
    Note that  if $p< q_\beta$ then  $\conc[F_\beta^{(\infty,1)}]$ evaluated at $p$ coincides with $F_\beta^{(\infty,1)}(p)$. This  corresponds to all trivial $P_{T|X}$ such that $\cp(Y|T) + \beta H(X|T) = \cp(Y) + \beta H(X)$. If, on the other hand, $p\geq q_\beta$, then $\conc[F_\beta^{(\infty,1)}$ is the convex combination of   $F_\beta^{(\infty,1)}(q_\beta)$ and $F_\beta^{(\infty,1)}(1-q_\beta)$. Hence, taking $q_\beta$ as a parameter (say, $\alpha$), the optimal binary $T^*$ is constructed as follows: $P_{X|T^*=0} = \sBer(\alpha)$ and $P_{X|T^*=1} = \sBer(\bar\alpha)$ for $\alpha\leq p$. Such channel induces \eq{\cp(Y|T^*) = \max\{\alpha*\delta, 1-\alpha*\delta\} = 1-\alpha*\delta,} as $\alpha\leq p\leq \frac{1}{2}$, and also 
    \eq{I(X;T^*) = h_\mathsf{b}(p) - h_\mathsf{b}(\alpha).}
    Combining these two, we obtain 
    \eq{\cp(Y|T^*) = 1-\delta*h^{-1}_\mathsf{b}(h_\mathsf{b}(p)-R).}
 \end{proof}

\begin{proof}[Proof of Theorem~\ref{Thm:Gerber_AlphaGamma}]
 Let $\sU_{\alpha}$ and $\sL_{\alpha}$ denote the $\sU_{\Phi,\Psi}$ and $\sL_{\Phi,\Psi}$, respectively, when $\Psi(Q_X) = \Phi(Q_X) = \|Q_X\|_\alpha$. In light of  \eqref{IB_alpha1} and \eqref{PF_alpha1}, it is sufficient to compute $\sL_{\alpha}$ and $\sU_{\alpha}$.
 To do so, we need to construct the lower convex envelope $\conv[F_\beta^{(\alpha)}]$ and upper concave envelope  $\conc[F_\beta^{(\alpha)}]$ of the map $F_\beta^{(\alpha)}(q)$ given by $q\mapsto \|Q_Y\|_\alpha-\beta\|Q_X\|_\alpha$ where 
 $X'\sim \sBer(q)$ and $Y'$ is the result of passing $X'$ through $\bsc(\delta)$, i.e., $Y'\sim \sBer(\delta*q)$. In this case, we have  
 \eqn{Prf:F_alpha}{q\mapsto F_\beta^{(\alpha)}(q) =  \|q*\delta\|_\alpha - \beta ||q||_\alpha,}
 where $\|a\|_\alpha$ is to mean $\|[a, \bar a]\|_\alpha$ for any $a\in [0,1]$.

 We begin by $\sL_{\alpha}$ for which we aim at obtaining $\conv[F_\beta^{(\alpha)}]$.  
 A straightforward computation shows that $F_\beta^{(\alpha)}(q)$ is convex for $\beta\leq (1-2\delta)^2$ and $\alpha\geq 2$. For $\beta> (1-2\delta)^2$ and $\alpha\geq 2$, it can be shown that  $F_\beta^{(\alpha)}(q)$ is concave  an interval $[q_\beta, 1-q_\beta]$ where $q_\beta$ solves $\frac{d}{dq}F_\beta^{(\alpha)}(q)=0$. (The shape of $q\mapsto F_\beta^{(\alpha)}(r)$ in is similar to what was depicted in Fig.~\ref{fig:Gerber}.) 
 By symmetry, $\conv[F_\beta^{(\alpha)}]$ is therefore obtained by replacing $F_\beta^{(\alpha)}(q)$ on this interval by $F_\beta^{(\alpha)}(q_\beta)$. Hence, if $p < q_\beta$, $\conv[F_\beta^{(\alpha)}]$ at $p$ coincides with $F_\beta^{(\alpha)}(p)$ which results in trivial $P_{T|X}$ (see the proof of Theorem~\ref{Thm:Gerber_Infty} for more details). If, on the other hand, $p\geq q_\beta$, then  $\conv[F_\beta^{(\alpha)}]$ evaluated at $p$ is given by a convex combination of $F_\beta^{(\alpha)}(q_\beta)$ and $F_\beta^{(\alpha)}(1-q_\beta)$.  Relabeling $q_\beta$ as a parameter (say, $q$), we can write an optimal binary $T^*$ via the following: $P_{X|T*=0}=\sBer(1-q)$ and $P_{X|T*=1}=\sBer(q)$ for $q\leq p$. This channel induces 
 $\Psi(Y|T^*) = \|q*\delta\|_\alpha$
 and 
 $\Phi(X|T^*) = \|q\|_\alpha$. Hence, the graph of $\sL_{\alpha}$ is given by  
 \eq{\big\{(\|q\|_\alpha, \|q*\delta\|_\alpha), 0\leq q\leq p\big\}.}  
 Therefore, 
 \eq{\sup_{\substack{P_{T|X}\\H_\alpha(X|T)\geq \zeta}}H_\alpha(Y|T) = \frac{\alpha}{1-\alpha}\log \|q*\delta\|_\alpha,} where $ q\leq p$ solves $\frac{\alpha}{1-\alpha}\log \|q\|_\alpha = \zeta$. Since the map $q\mapsto \|q\|_\alpha$ is strictly decreasing for $q\in [0,0.5]$, this equation has a unique solution.
 
Next, we compute $\sU_{\alpha}$ or equivalently $\conc[F_\beta^{(\alpha)}]$ the upper concave envelop of $F_\beta^{(\alpha)}$ defined in \eqref{Prf:F_alpha}. As mentioned earlier, $q\mapsto F_\beta^{(\alpha)}(q)$ is convex for $\beta \leq (1-2\delta)^2$ and $\alpha\geq 2$. For $\beta > (1-2\delta)^2$, we need to consider three cases: (1) $\conc[F_\beta^{(\alpha)}]$ is given by the convex combination of $F_\beta^{(\alpha)}(0)$ and $F_\beta^{(\alpha)}(1)$, (2) $\conc[F_\beta^{(\alpha)}]$ is given by the convex combination of $F_\beta^{(\alpha)}(0)$, $F_\beta^{(\alpha)}(\frac{1}{2})$, and $F_\beta^{(\alpha)}(1)$, (3) $\conc[F_\beta^{(\alpha)}]$ is given by the convex combination of $F_\beta^{(\alpha)}(0)$ and $F_\beta^{(\alpha)}(q^\dagger)$ where $q^\dagger$ is a point $\in [0,\frac{1}{2}]$.   Without loss of generality, we can ignore the first case. The other two cases correspond to the following solutions
\begin{itemize}
    \item $T^*$ is a ternary variable given by $P_{X|T^*=0} = \sBer(0)$, $P_{X|T^*=1} = \sBer(1)$, and $P_{X|T^*=2} = \sBer(\frac{1}{2})$ with marginal $T^*\sim \sBer(1-p-\frac{\lambda}{2}, p-\frac{\lambda}{2}, \lambda)$ for some $\lambda\in[0,2p]$. This produces
    \eq{\Psi(Y|T^*) = \bar\lambda \|\delta\|_\alpha + \lambda \|\frac{1}{2}\|_\alpha,}
    and 
    \eq{\Phi(X|T^*) = \bar\lambda + \lambda \|\frac{1}{2}\|_\alpha.}
    \item $T^*$ is a binary variable given by $P_{X|T^*=0} = \sBer(0)$ and $P_{X|T^*=1} = \sBer(\frac{p}{\lambda})$ with marginal $T^*\sim \sBer(\lambda)$ for some $\lambda\in[2p,1]$. This produces 
    \eq{\Psi(Y|T^*) = \bar\lambda \|\delta\|_\alpha + \lambda\|\delta*\frac{p}{\lambda}\|_\alpha,}
    and 
    \eq{\Phi(X|T^*) = \bar\lambda + \lambda\|\frac{p}{\lambda}\|_\alpha.}
\end{itemize}
  Combining these two cases, can write 
  \eq{\sU_\alpha (\zeta) = \bar\lambda\|\delta\|_\alpha + \lambda\|\frac{q}{z}*\delta\|_\alpha,}
  where 
  \eq{\zeta = \bar\lambda + \lambda\|\frac{q}{z}\|_\alpha,}
  and $z = \max\{2p, \lambda\}$.
 Plugging this into \eqref{IB_alpha1} completes the proof.
\end{proof}

\begin{proof}[Proof of Lemma~\ref{Lemma:Alpha_Gamma_infinity}]
 The facts that $\gamma\mapsto H_\gamma(X|T)$ is non-increasing on $[1, \infty]$ \cite[Proposition 5]{Arimoto's_Conditional_Entropy} and $\left(\sum_{i}|x_i|^\gamma\right)^{1/\gamma}\geq \max_{i}|x_i|$ for all $p\geq 0$ imply
 \eqn{Proof_sandwich}{\frac{\gamma-1}{\gamma}H_\gamma(X|T)\leq H_\infty(X|T)\leq H_\gamma(X|T).}
 Since $I_\infty(X;T) = H_\infty(X) - H_\infty(X|T)$, the above lower bound yields
 \eqn{PF_Bound1}{I_\gamma(X; T)\geq \frac{\gamma}{\gamma-1}I_\infty(X; T)-\frac{\gamma}{\gamma-1}H_\infty(X) + H_\gamma(X),
 }
 where the last inequality follows from the fact that $\gamma\mapsto H_\gamma(X)$ is non-increasing. The upper bound in \eqref{Proof_sandwich} (after replacing $X$ with $Y$ and $\gamma$ with $\alpha$) implies
 \eqn{PF_Bound2}{I_\alpha(Y; T)\leq I_\infty(Y; T) + H_\alpha(Y) - H_\infty(Y).}
 Combining \eqref{PF_Bound1} and \eqref{PF_Bound2}, we obtain the desired upper bound for $\pf^{(\alpha, \gamma)}$. The other bounds can be proved similarly by interchanging $X$ with $Y$ and $\alpha$ with $\gamma$ in \eqref{PF_Bound1} and \eqref{PF_Bound2}.
\end{proof}

\bibliographystyle{IEEEtran}
\bibliography{references}
\end{document}